\documentclass[journal,onecolumn]{IEEEtran}
\usepackage{cite}
\usepackage{amsmath}
\usepackage{mathtools}
\usepackage{amsfonts,amsthm,amssymb, bbm}
\usepackage{thmtools, thm-restate}
\usepackage{dsfont}
\usepackage{mymacros}
\usepackage{fancyhdr}
\usepackage{datetime}
\usepackage{color}
\usepackage{enumerate}
\usepackage{subcaption}
\usepackage{hyperref}
\usepackage[capitalize]{cleveref}
\crefname{equation}{}{}

\usepackage{tikz}
\usetikzlibrary{cd}
\usetikzlibrary{dsp}
\usetikzlibrary{decorations.pathreplacing}

\usepackage{graphicx}
\usepackage[outdir=./fig/converted/]{epstopdf}
\graphicspath{{../eps/}}
\DeclareGraphicsExtensions{.eps,pdf}

\newcommand{\ThetaNd}[0]{\Theta^d_{N}}
\newcommand{\ThetaNt}[0]{\Theta^2_{N}}
\newcommand{\kapt}[0]{\kappa_N}

\newcommand{\n}[0]{\mathbf{n}}
\newcommand{\NN}[0]{\left[N\right]}
\newcommand{\Pn}[0]{\norm{p}_{N^d, \infty}}    %
\newcommand{\A}[0]{\Pn}  % Upper bound over thetaNd
\newcommand{\B}[0]{B}    % Lower bound over thetaNd
\newcommand{\Cnd}[0]{C_{N,n}^d}  
\newcommand{\tv}[0]{\omega}
\newcommand{\ti}[0]{\tv_i}
\newcommand{\tki}[0]{\tv_{k_i}}
\newcommand{\tk}[0]{\tv_k}

\usepackage{xspace}
\newcommand*{\eg}{\textit{e}.\textit{g}.\@\xspace}
\newcommand*{\ie}{\textit{i}.\textit{e}.\@\xspace}
\newcommand*{\etal}{\textit{et al}.\@\xspace}

\newcommand{\dvp}{de la Vall\'{e}e-Poussin }

\usepackage{todonotes}

\newcommand{\ceil}[1]{\mathrm{ceil}(#1)}

\declaretheorem[name=Lemma, refname={Lemma,Lemmas}]{lemma}

\begin{document}
\title{Bounding Multivariate Trigonometric Polynomials with Applications to Filter Bank Design}

\author{Luke~Pfister,~\IEEEmembership{Student Member,~IEEE,}
        Yoram~Bresler,~\IEEEmembership{Fellow,~IEEE}%
\thanks{
  This work was supported in part by
 the National Science Foundation (NSF) under grant CCF-1320953.
}}% <-this % stops a space

\maketitle
\begin{abstract}
  The extremal values of multivariate trigonometric polynomials are of interest in fields ranging
  from control theory to filter design, but finding the extremal values of such a polynomial is
  generally NP-Hard. In this paper, we develop simple and efficiently computable estimates of the
  extremal values of a multivariate trigonometric polynomial directly from its samples. We provide
  an upper bound on the modulus of a complex trigonometric polynomial, and develop upper and lower
  bounds for real trigonometric polynomials. For a univarite polynomial, these bounds are tighter
  than existing bounds, and the extension to multivariate polynomials is new. As an application, the
  lower bound provides a sufficient condition to certify global positivity of a real trigonometric
  polynomial. We use this condition to motivate a new algorithm for multi-dimensional, multirate,
  perfect reconstruction filter bank design. We demonstrate our algorithm by designing a 2D perfect
  reconstruction filter bank.
\end{abstract}

% note that keywords are not normally used for peerreview papers.
% \begin{IEEEkeywords}
% IEEEtran, journal, \LaTeX, paper, template.
% \end{IEEEkeywords}

% \begin{table}[ht]
%   \centering
%   \begin{tabular}{| c | c |} \hline
%     $n$ & Polynomial degree \\
%     \hline $d$ & Polynomial dimension \\
%     \hline
%   \end{tabular}
% \end{table} 
\section{Introduction}
\subsection{Motivation}
Trigonometric polynomials are intimately linked to discrete-time signal processing, arising 
in problems of controls, communications, filter design,  and super
resolution, among others. For example, the Discrete-Time Fourier Transform
(DTFT) converts a sequence of length $n$ into a trigonometric polynomial of degree $n-1$.
Multivariate trigonometric polynomials arise in a similar fashion, as the $d$-dimensional DTFT
yields a $d$-variate trigonometric polynomial.

The extremal values of a trigonometric polynomial are often of interest. In an Orthogonal Frequency
Division Multiplexing (OFDM) communication system, the transmitted signal is a univariate
trigonometric polynomial, and the maximum modulus of this signal must be accounted for when
designing power amplifiers \cite{Jetter2001}. The maximum modulus of a trigonometric polynomial is
related to the stability of a control system in the face of perturbations \cite{Dumitrescu2017}. The
maximum gain and attenuation of a Finite Impulse Response (FIR) filter are the maximum and minimum
values of a real and non-negative trigonometric polynomial. Unfortunately, determining the extremal
values of a multivariate polynomial given its coefficients is NP-Hard \cite{Murty1987, Parrilo2003}.

An approximation to the extremal values can be found by discretizing the polynomial and performing a
grid search, but this method is sensitive to the discretization level. Instead, one can try to find
the extremal values using an optimization-based approach. However, iterative descent algorithms are
prone to finding local optima as a generic polynomial is not a convex function. The sum-of-squares
machinery provides an alternative approach: extremal values of a polynomial can be found by solving
a hierarchy of semidefinite program (SDP) feasibility problems \cite{Parrilo2001, Parrilo2003,
  Dumitrescu2017}. Truncating the sequence of SDPs provides a lower (or upper) bound to the minimum
(or maximum) of the polynomial. However, the size of the SDPs grows exponentially in the number of
variables, $d$, and polynomially in the degree, $n$, limiting the applicability of this approach.

In many applications we have access to samples of the polynomial rather than to the coefficients of
the polynomial itself. Equally spaced samples of a trigonometric polynomial arise, for instance,
when computing the Discrete Fourier Transform (DFT) of a sequence. Given enough samples, the
polynomial can be evaluated at any point by periodic interpolation, and thus grid search or
optimization-based approaches can still be used; however, the previously
described issues of discretization error, local minima, and complexity remain.

In this paper, we derive simple estimates for the extremal values of a multivariate trigonometric
polynomial directly from its samples, \ie with no interpolation step. For a complex polynomial we
provide an upper bound on its modulus, while for a real trigonometric polynomial we provide upper
and lower bounds. Upper bounds of this style have been derived for univariate trigonometric
polynomials-- our work provides an extension to the multivariate case. We describe two sample
applications that benefit from our lower bound and from the extension to multivariate polynomials.

i) \textbf{Design of Perfect Reconstruction Filter Banks}. A multi-rate filter bank in $d$ dimensions
is characterized by its polyphase matrix, $H(z) \in \Cbb^{m \times n}$, where each entry in the
matrix is a $d$-variate Laurent polynomial\footnote{A Laurent polynomial allows negative
  powers of the argument.} in $z \in \Cbb^d$ \cite{Vetterli1987}.

Many important properties of the filter bank can be inferred from the polyphase matrix. A filter
bank is said to be \emph{perfect reconstruction} (PR) if any signal can be recovered, up to scaling
and a shift, from its filtered form. The design and characterization of multirate filter banks in
one dimension is well understood, but becomes difficult in higher dimensions due to the lack of a
spectral factorization theorem \cite{Vaidyanathan1992, Do2011a, Venkataraman1994, Zhou2005,
  Delgosha2004}.

The perfect reconstruction condition is equivalent to 
the strict positivity of the real trigonometric polynomial $p_H(\omega) = \det \left( H^*(e^{j\omega})
  H(e^{j\omega}) \right)$ \cite{Vetterli1987, Cvetkovic1998}.   In 
\cref{sec:filter_design}, we use our lower bound to develop a new algorithm for the design and
construction of multi-dimensional perfect reconstruction filter banks.

ii) \textbf{Estimating the smallest eigenvalue of a Hermitian Block Toeplitz matrix with
  Toeplitz Blocks.}

Toeplitz matrices describe shift-invariant phenomena and are found in countless
applications. Toeplitz matrices model convolution with a finite impulse response filter, and the
covariance matrix formed from a random vector drawn from a wide-sense stationary (WSS) random
process is symmetric and Toeplitz. An $n \times n$ Toeplitz matrix is of the form
\begin{equation}
  X_n =
  \begin{bmatrix}
    x_0    & x_{-1} & x_{-2}  & \cdots & x_{-n+1} \\
    x_1    & x_0    & x_{-1} &         &           \\
    x_2    & x_1    & x_0    &         & \vdots    \\
    \vdots &        &        & \ddots  &           \\
    x_{n-1} &       &        & \cdots  & x_0
  \end{bmatrix},
\end{equation}
and a Hermitian symmetric Toeplitz matrix satisfies $x^*_{i} = x_{-i}$. Associated with $X_n$ is the
trigonometric polynomial \footnote{This differs from the usual approach of describing Toeplitz
  matrices, wherein a Toeplitz matrix of size $n$ is generated according to \eqref{eq:toep_synth}
  for an underlying symbol and the behavior as $n\to \infty$ is investigated. Here, we work with a
  Toeplitz matrix of fixed size.}
\begin{equation}
  \label{eq:symbol}
  \hat{x}(\tv) = \sum_{k=-n}^{n} x_k e^{j \tv k}, \quad \quad -\pi \leq \tv < \pi,
\end{equation}
with coefficients
\begin{equation}
  \label{eq:toep_synth}
  x_k = \frac{1}{2\pi}\int_{-\pi}^{\pi} \hat{x}(\tv) e^{-j k \tv} dt, \quad \quad k \in \Zbb.
\end{equation}
The polynomial $\hat{x}$ is known as the \emph{symbol} of $X_n$. If the symbol is real then $X_n$ is
Hermitian, and if $\hat{x}$ is strictly positive then $X_n$ is positive definite.

A vast array of literature has examined the connections between a real symbol $\hat{x}$ and the
eigenvalues of the Hermitian Toeplitz matrices $X_n$ as $n \to \infty$; see \cite{Albrecht1999,
  Gray2005} and references therein. One result of particular interest states that the eigenvalues of
$X_n$ are upper and lower bounded by the supremum and infimum of the symbol.

The smallest eigenvalue of a Toeplitz matrix is of interest in many applications\cite{Chan1994a,
  Chan1994, Pisarenko1973}, and there are several iterative algorithms to efficiently calculate this eigenvalue
\cite{Laudadio2008}. We propose a non-iterative estimate of the smallest and largest eigenvalues of $X_n$ 
by first bounding the eigenvalues in terms of the symbol, then bounding the symbol
in terms of the entries of $X_n$.

Shift invariant phenomena in two dimensions are described by Block Toeplitz matrices with Toeplitz
Blocks (BTTB). The symbol for a BTTB matrix is a bi-variate trigonometric polynomial, and the
bounds developed in this paper hold in this case.

\subsection{Notation}
For a set $\mathbb{X}$, let $\mathbb{X}^d$ be the $d$-fold Cartesian product $\mathbb{X} \times
\hdots \times \mathbb{X}$. Let $\Tbb = [0, 2\pi]$ be the torus and $\Zbb$ be the integers. The set
$\left\{0, \hdots N-1 \right\}$ is written $\NN$. We denote the space of $d$-variate trigonometric
polynomials with maximum component degree $n$ as
\begin{equation}
  T^d_n \triangleq \spanset{ e^{j k \cdot \tv} : \tv \in \Tbb^d,  k \in \Zbb^d, \norm{k}_\infty \leq  n},
\end{equation}
where $x \cdot y \triangleq \sum_{i=1}^d x_i y_i$ is the Euclidean inner product and $\norm{k}_\infty =
\max_{1 \leq i \leq d} \abs{k_i}$.
An element of $T_n^d$ is explicitly given by 
\begin{equation}
  p(\tv) = \sum_{k_1=-n}^{n} \hdots \sum_{k_d=-n}^n c_{k_1 \hdots k_d} e^{j k_1 \tv_1} \hdots e^{j k_d \tv_d}.
\end{equation}
If the coefficients satisfy $c_{k_1,\hdots,k_d} = c^*_{-k_1,\hdots,-k_d}$, then $p(\tv)$
is real for all $\tv$ and $p$ is said to be a \emph{real trigonometric polynomial}. We denote the
space of real trigonometric polynomials by $\bar{T}_n^d$.  For $p \in T_n^d$ let $\norm{p}_\infty =
\max_{\tv \in \Tbb^d} \abs{p(\tv)}$.
We write the set of $N$ equidistant sampling points on $\Tbb$ as 
\begin{equation}
  \Theta_N \triangleq \left\{\tk = k \frac{2 \pi}{N} : k = 0, \hdots, N-1 \right\}, 
\end{equation} and on $\Tbb^d$ as $\ThetaNd$, given by the $d$-fold Cartesian product $\Theta_N
\times \hdots \times \Theta_N$. The maximum modulus of $p$ over $\ThetaNd$ is
\begin{equation}
  \norm{p}_{N^d, \infty} \triangleq \max_{\tv \in \ThetaNd} \abs{p(\tv)}.
\end{equation}

\subsection{Problem Statement and Existing Results}
Let $p \in \bar{T}_n^d$. Our goal is to find scalars $a \leq b$, depending only on $N, d$, and the
$N^d$ samples $\left\{ p(\tv) : \tv \in \ThetaNd \right\}$, such that
\begin{equation}
  \label{eq:desired_bound}
  a \leq p(\tv) \leq b.
\end{equation}
For complex trigonometric polynomials, $p \in T_n^d$, we want an upper bound on
the modulus; a lower bound on the modulus can be obtained by considering the real
  trigonometric polynomial $p^{\prime} \in \bar{T}_n^{2d}: \tv \mapsto \abs{p(\tv)}^{2}$.

  By the periodic sampling theorem (\cref{thm:dirichlet_interp}), trigonometric interpolation
  perfectly recovers $p \in T_n^d$ from $(2n+1)^d$ uniformly spaced samples. A standard result of
  approximation theory states \cite{Zygmund2005, Soerevik2015}
\begin{equation}
  \label{eq:lebesgue_bound}
  \norm{p}_\infty \leq \norm{p}_{(2n+1)^d, \infty} \left( \frac{\pi + 4}{\pi} + \frac{2}{\pi}\log(2n + 1)  \right)^d, 
\end{equation}
but this becomes weak as the polynomial degree $n$ or the dimension $d$ of its domain increases. A
more stable estimate is obtained by using non-uniformly spaced samples. However, in many
applications the sampled polynomial is obtained using the DFT, thus providing uniformly spaced
samples.

Our aim is to get stronger estimates by using more (uniformly spaced) samples than are required by
the periodic sampling theorem. Upper bounds for \emph{univariate} trigonometric polynomials have been
developed using this strategy. Let $p \in T_n$. Given an integer $m$ and $N = 2 m > 2n+1$ samples
of $p$, Ehlich and Zeller showed
\begin{equation}
  \label{eq:ehlich_bound}
  \norm{p}_\infty \leq \left(\cos{\left( \frac{\pi n}{2 m} \right)}  \right)^{-1} \norm{p}_{N, \infty}
\end{equation}
and this bound is sharp if $n$ is a divisor of $m$.

Wunder and Boche developed a more flexible bound:  given $N \geq 2n +1$, they showed \cite{Wunder2002}
\begin{equation}
  \label{eq:wunder_bound}
  \norm{p}_\infty \leq \sqrt{\frac{N + 2n + 1}{N - (2n + 1)}} \norm{p}_{N, \infty}.
\end{equation}
Zimmermann \etal refined this bound to 
\begin{equation}
  \label{eq:crest_factor_1d}
  \norm{p}_\infty \leq \frac{\norm{p}_{N, \infty} }{\sqrt{1 - \alpha}},
\end{equation}
where $\alpha = 2n / N$.  
The quantity $\alpha^{-1}$ is almost equal to the oversampling factor 
$\frac{N}{2n + 1}$, and plays the same role:  $\alpha$ is a decreasing function of $N$, and for $N
\geq 2n +1$, we have $\alpha < 1$.

The bounds \cref{eq:lebesgue_bound,eq:ehlich_bound,eq:wunder_bound,eq:crest_factor_1d} each have the
form:
\begin{equation}
  \norm{p}_\infty \leq \Cnd \A, 
  \label{eq:operator_norm_ub}
\end{equation}
where $\Cnd$ is a real, non-negative constant that depends on $N, n$ and, in the case of
\cref{eq:lebesgue_bound}, $d$.
In the univariate case, Zimmermann \etal studied the optimal value of $C_{N, n}$ and showed that it
depends only on $N/n$. They also characterized \emph{extremal} polynomials, for which
\cref{eq:operator_norm_ub} holds with equality, and discussed a Remez-like algorithm to construct
such polynomials for given $N$ and $n$ \cite{Jetter2001}.

\subsection{Contributions}
Our contributions can be summarized as follows: (i) we develop upper bounds of the form
\cref{eq:operator_norm_ub} for \emph{multivariate} trigonometric polynomials; these include both a
multivariate extension of
the bound \eqref{eq:crest_factor_1d}, as well as a tighter bound for the case of low oversampling
($N \approx 2n+1$);
(ii) we specialize and strengthen the bounds for real polynomials;
(iii) we derive a lower bound for real trigonometric polynomials;
and (iv) we apply our bounds to the design of multi-dimensional perfect-reconstruction
filter banks.

%%% Local Variables:
%%% mode: latex
%%% TeX-master: "trig_poly_main"
%%% End:

\section{Statement of Main Results}
In this section we collect our main results; proofs are deferred to
\cref{sec:complex_proofs,sec:real_proofs}. For simplicity we work with $T_n^d$, but the results can
be easily strengthened by allowing for the component degree to vary in each of the $d$ dimensions.

Our first task is to obtain bounds of the form \cref{eq:operator_norm_ub} for multivariate
trigonometric polynomials. We have a pair of such bounds:
\begin{restatable}{theorem}{crestbound}
  \label{thm:crest_bound}
  Let $p \in T_n^d$. Take $N \geq 2n + 1$ and set $\alpha = 2n / N$.
  Then
  \begin{equation}
    \label{eq:upper_bound_general}
    \norm{p}_\infty \leq \Cnd \A, 
  \end{equation}
  where
  \begin{align}
    \Cnd  &\triangleq
      \frac{
      \displaystyle
      \left(
    \sup_{\tv \in \Tbb}
      \left\{  
    \sum_{\tk \in \Theta_N}
            \abs{
    \frac{
            \sin{\left( \frac{N \tv}{2}\right)} 
            \sin{\left(\frac{N - 2n}{2} (\tv - \tk)\right)}}
            {\sin^2{( (\tv - \tk) / 2)}}}
            \right\}
      \right)^d
      }{N^d (N - 2n)^d} \label{eq:crest_bound_unwieldy} \\
    &\leq \left( 1 - \alpha \right)^{-\frac{d}{2}} \label{eq:crest_bound}.
  \end{align}
  Further, $\Cnd \A - \norm{p}_\infty = \mathcal{O}(dn/N)$.
\end{restatable}
The bound \cref{eq:crest_bound_unwieldy}
involves only a univariate function and can be 
calculated numerically.  Still, the expression is unwieldy; \cref{eq:crest_bound} is a
simpler, but weaker, alternative.  

We plot the behavior of $C_{N, n}$, given by \cref{eq:crest_bound_unwieldy} and
\cref{eq:crest_bound} for the $d=1$ univariate case, in \cref{fig:operator_norm}. Also shown in
\cref{fig:operator_norm} are the optimal values of $C_{N, n}$ for integer oversampling factors,
given by \cref{eq:ehlich_bound}, and the values obtained using Zimmermann's Remez-like algorithm
\cite{Jetter2001}.

The upper bound \cref{eq:upper_bound_general} with $\Cnd$ given by \cref{eq:crest_bound_unwieldy} is
nearly tight for $N / (2n) < 2$, whereas replacing $\Cnd$ by its upper bound \cref{eq:crest_bound}
results in a weakening of \cref{eq:upper_bound_general} in this regime. This gap makes
\cref{eq:crest_bound_unwieldy} particularly attractive in the $d$-variate case, where the bounds are
raised to the $d$-th power, further increasing the gap between \cref{eq:crest_bound_unwieldy} and \cref{eq:crest_bound}.

However, for oversampling factor greater than two, \ie $N / (2n) > 2$, the difference in using
\cref{eq:crest_bound_unwieldy} or \cref{eq:crest_bound} becomes negligible. Both
bounds coincide with the optimal value at $N =
4n$, and are within roughly $10\%$ of the optimal value for large oversampling factors. Hence, both
\cref{eq:crest_bound_unwieldy} and \cref{eq:crest_bound} are useful, in different oversampling
regimes.

\begin{figure}[t]
  \centering
    \includegraphics[width=3.5in]{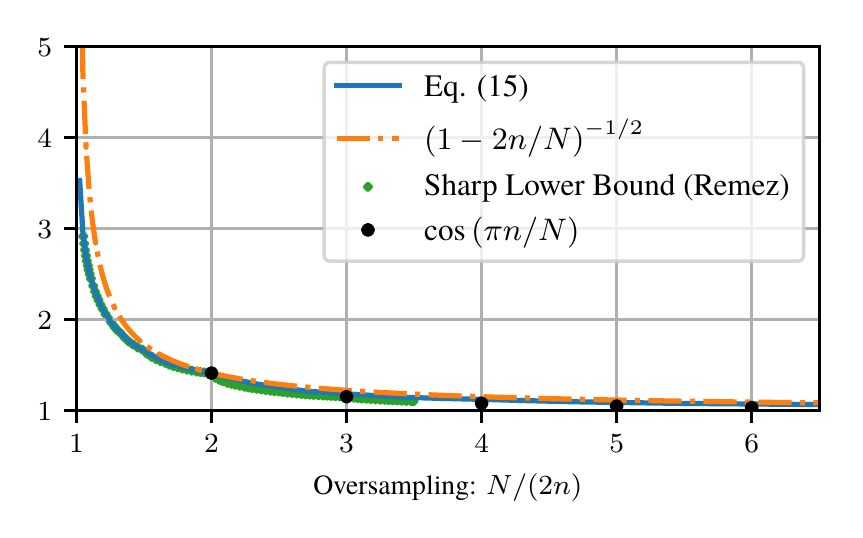}
    \includegraphics[width=3.5in]{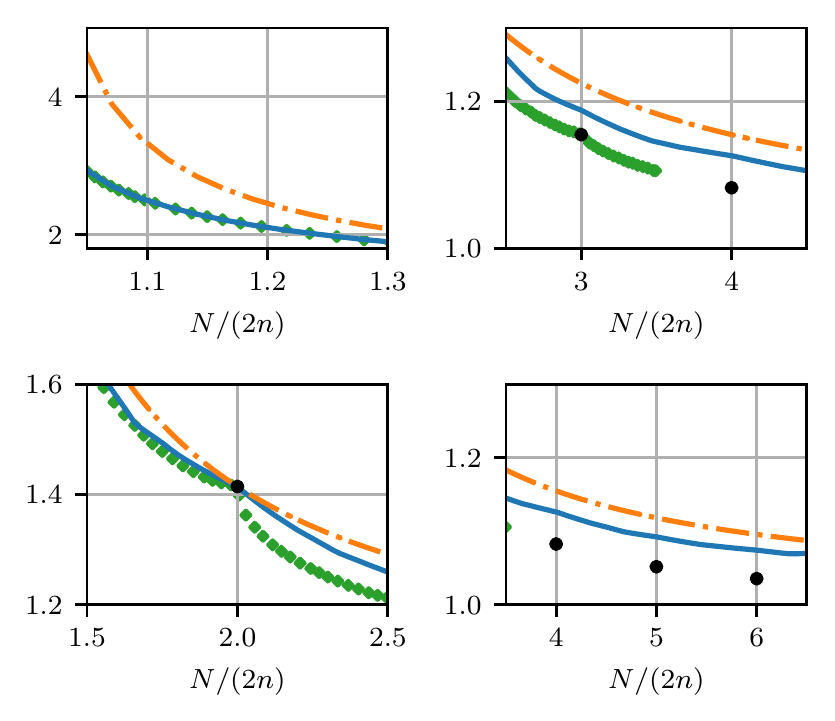}
    \caption{Comparing upper bounds of the form \cref{eq:upper_bound_general} as a function of
      oversampling ratio, $N / 2n$. Green diamonds indicate the optimal upper bound as calculated
      using a Remez-type algorithm \cite[Fig. 2]{Jetter2001}. Black dots denote the
      tight upper bound \cref{eq:ehlich_bound} occuring at integer
      oversampling ratios $N=m n \geq 2n+1$.}
    \label{fig:operator_norm} 
\end{figure}

Next, we obtain a tighter estimate by restricting our attention to real polynomials.
\begin{restatable}{corollary}{crestboundnew}
  \label{thm:crest_bound_new}
  Let $p \in \bar{T}_n^d$ and take $N \geq 2n + 1$.
  Set $A \triangleq \max_{\tv \in \ThetaNd} p(\tv)$,  $B \triangleq \min_{\tv \in \ThetaNd} p(\tv)$ and take
  $\Cnd$ as in \cref{thm:crest_bound}.  Then,
  \begin{equation}
    \label{eq:crest_bound_new}
    \norm{p}_\infty
    \leq \frac{1}{2}\left(A + B + \Cnd \left(A - B\right) \right).
  \end{equation}
\end{restatable}
The estimate \eqref{eq:crest_bound_new} coincides with \eqref{eq:upper_bound_general} in the case that
$\min_{\tv\in\ThetaNd} p(\tv) = - \max_{\tv\in\ThetaNd} p(\tv)$, and is tighter otherwise, making
this refinement especially useful for non-negative polynomials.

Using \cref{thm:crest_bound_new} we obtain the following lower bound:
\begin{restatable}{corollary}{lowerbound}
  \label{thm:lower_bound}
  Let $p \in \bar{T}_n^d$ and take $N \geq 2n + 1$. 
  Set $A \triangleq \max_{\tv \in \ThetaNd} p(\tv)$,  $B \triangleq \min_{\tv \in \ThetaNd} p(\tv)$ and take
  $\Cnd$ as in \cref{thm:crest_bound}.
  Then, for all $\tv \in \Tbb^d$, 
  \begin{equation}
    \label{eq:lower_bound}
    p(\tv) \geq \frac{1}{2}\left(A  + B - \Cnd \left(A - B\right)  \right).
  \end{equation}
\end{restatable}

By \cref{thm:crest_bound}, $\Cnd \to 1$ as $N \to \infty$. Thus as $N \to \infty$, the right
hand side of \cref{eq:lower_bound} approaches $B$, and by continuity we have $B = \min_{\tv \in
  \ThetaNd} p(\tv) \to \min_{\tv \in \Tbb^d} p(\tv)$. Thus the bound is tight as $N \to \infty$. In the
case of $A = B$, the right hand side of \cref{eq:lower_bound} is $A = \Pn$, and thus $p(\tv) > 0$ so
long as the samples of $p$ are not uniformly zero. This is expected, as otherwise the polynomial
$p(\tv) - \Pn \in T_n^d$ would vanish on a set of $N^d > (2n +1)^d$ points, which is impossible unless
the polynomial is identically zero.

A little algebra on \cref{eq:lower_bound} establishes a sufficient condition to verify the strict
positivity of a multivariate trigonometric polynomial.
\begin{restatable}{corollary}{noberncondition}
  \label{thm:lower_bound_condition}
  Let $p \in \bar{T}_n^d$ and $N \geq 2n + 1$.
  Set $\alpha = 2n / N$.
  If $p(\tv) > 0$ for all $\tv \in \ThetaNd$ and
  \begin{equation}
    \label{eq:lower_bound_cond_cnd}
    \kapt \triangleq \frac{\max_{\tv \in \ThetaNd} p(\tv)}{\min_{\tv \in \ThetaNd} p(\tv)} 
   \leq \frac{\Cnd + 1}{\Cnd - 1} 
  \end{equation}
  then $p(\tv) > 0$ for all $\tv \in \Tbb^d$.
  Furthermore, as $\Cnd \leq \left( 1 - \alpha \right)^{-\frac{d}{2}}$,
  \cref{eq:lower_bound_cond_cnd} 
  can be replaced by the more stringent, but easier to evaluate, condition
  \begin{equation}
    \label{eq:lower_bound_cond_alpha}
  \kapt \leq \frac{1 + (1 - \alpha)^{\frac{d}{2}}}{ 1 - (1 - \alpha)^{\frac{d}{2}}}.
  \end{equation}
\end{restatable}
For $p \in \bar{T}_n^d$ with non-negative samples, we call the quantity $\kapt$ in
\cref{eq:lower_bound_cond_cnd} the \emph{N-sample dynamic range}; here, $0 / 0$ is taken to be $1$.

\cref{thm:lower_bound_condition} provides an easy way to certify strict positivity of a real,
non-negative polynomial from its samples: simply calculate the dynamic range $\kapt$ and verify that
\cref{eq:lower_bound_cond_cnd} or \cref{eq:lower_bound_cond_alpha} holds. These conditions are
easier to satisfy (as a function of the oversampling rate) for polynomials whose maximum and minimum
sampled values are close to one another. Intuitively, if the sampled values of a real trigonometric
polynomial are strictly positive and don't vary ``too much'', then the polynomial is strictly
positive over its entire domain. For fixed $n$ and $d$, the right hand sides of
\cref{eq:lower_bound_cond_cnd} and \cref{eq:lower_bound_cond_alpha} are increasing functions of $N$,
illustrating a tradeoff: polynomials with a large amount of variation, and thus large values of
$\kapt$, require larger oversampling factors $N$ for the bounds to hold. Note that $\kapt$ is not
necessarily a monotone function of $N$, but is monotone in $k$ when choosing $N=2^k$.
\cref{fig:certificate1} illustrates the regions for which
\cref{eq:lower_bound_cond_cnd,eq:lower_bound_cond_alpha} hold. In \cref{sec:filter_design} we use
this condition to inform the design of multidimensional perfect reconstruction filter banks.

\begin{figure}
  \centering
    \centering
    \includegraphics[width=3.5in]{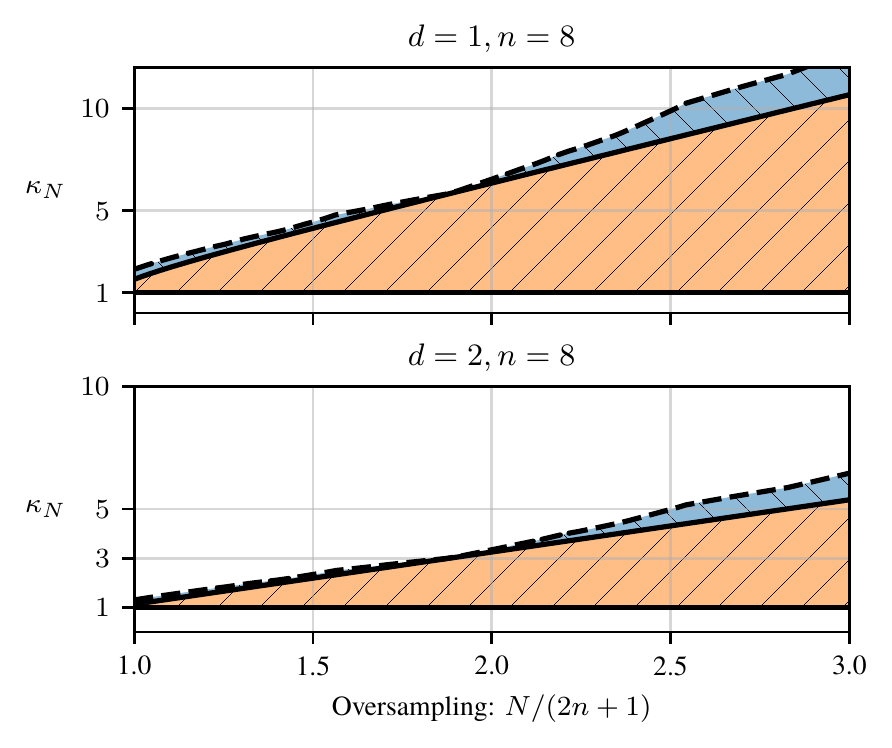}
\caption{
  Any $p \in \bar{T}_n^d$ with positive samples and whose $N$-sample signed dynamic range
  $\kapt$ lies in the shaded region must be strictly positive.
  The orange shaded region is certified
  using \cref{eq:lower_bound_cond_alpha}, while the blue region uses \cref{eq:lower_bound_cond_cnd}.
  }
  \label{fig:certificate1}
\end{figure}

\section{Proof of \cref{thm:crest_bound}}
\label{sec:complex_proofs}
We begin by proving \cref{thm:crest_bound}, which extends the upper bound \eqref{eq:crest_factor_1d}
from univariate to multivariate polynomials and provides a tighter result for the case of low
oversampling. As $T_n^d$ is constructed as the $d$-fold tensor product of $T_n$ with itself, the
proof is similar to the follows the univariate case \cite{Jetter2001}. We consider both real and
complex trigonometric polynomials.

\subsection{Interpolation by the Dirichlet Kernel}
\label{sub:dirichlet_interpolation}
For $\n = [n_1, \hdots n_d] \in \NN^d$, the $\n$-th order Dirichlet kernel is the tensor product of $d$ 
kernels, each of order $n_i$:
\begin{equation}
  D_{\n}^d(\tv) \triangleq \sum_{\abs{k_i} \leq n_i} e^{j k \cdot \tv} =
  \prod_{i=1}^d \frac{\sin{\frac{2 n_i +1}{2}\ti}}{\sin{\frac{\ti}{2}}}
  \quad 
  \tv \in \Tbb^d, k \in \Zbb^d.
  \label{eq:dirichlet_kernel}
\end{equation}
If $\n$ is identical in each index (\ie $n_i = n$ for each $i \in \left[d\right]$) we write the
kernel as $D_n^d(\tv)$.  The Dirichlet kernel is key to the periodic sampling formula:
\begin{lemma}
 \label{thm:dirichlet_interp}
 Let $p \in T_n^d$ be sampled on $\ThetaNd$. Let $m$ be an integer with $m > n$. If $N \geq n + m$,
 then
  \begin{equation}
    \label{eq:dirichlet_interp}
    p(\tv) = \frac{1}{N^d} \sum_{\tk \in \ThetaNd} p(\tk) D^d_m(\tv - \tk) 
  \end{equation}
  for all $\tv \in \Tbb^d$.
\end{lemma}
\cref{thm:dirichlet_interp} (\eg, \cite{Proakis2006}) is the periodic counterpart of sinc interpolation arising in the
Whittaker-Shannon interpolation formula.
The bound \eqref{eq:lebesgue_bound} can be obtained from \eqref{eq:dirichlet_interp} when 
$N = 2n + 1$ \cite{Soerevik2015}.    

\subsection{Interpolation by the \dvp Kernel}
A better result is obtained by oversampling $(N > 2n + 1)$ and exploiting the nice properties
of summation kernels.

Let $n, m$ be integers with $m > n$ and define $\mathbb{V}_{n,m}^d = \left\{ l \in \Zbb^d : n \leq l_i < m \right\}$. The $n, m$-th \dvp kernel is defined as the moving average of Dirichlet kernels:
\begin{align}
  D^d_{n,m}(\tv) &\triangleq \frac{1}{(m-n)^d} \sum_{\ \ \mathclap {\n \in \mathbb{V}^d_{n,m}}} D^d_{\n}(\tv) 
  \label{eq:vp_definition} \\
  &= \frac{1}{(m-n)^d} \prod_{i=1}^d \frac{ \sin{(\frac{m + n}{2} \ti)}\sin{(\frac{m - n}{2} \ti)}}{ \sin^2{(\ti/2)}}\label{eq:vp_kernel_factored}.
\end{align}
Taking $n=0$ recovers the well-known Fej\'{e}r kernel \cite{Stein2003},
\begin{equation}
  D^d_{0, m} = \frac{1}{m^d} \prod_{i=1}^d \frac{ \sin^2{(\frac{m}{2} \ti)}}{\sin^2{(\ti/2)}}.
  \label{eq:fejer_kernel}
\end{equation}
The Fej\'{e}r kernel is used to derive the bound \cref{eq:wunder_bound} \cite{Wunder2002}.

Importantly, the \dvp kernel inherits the reproducing property of the Dirichlet kernel.
\begin{lemma}
  \label{thm:vp_interp}
  For any $p \in T_n^d$ we have
  \begin{equation}
    p(\tv) = \frac{1}{N^d} \sum_{\tk \in \ThetaNd} p(\tk) D^d_{n,m}(\tv - \tk)
  \end{equation}
  for all $\tv \in \Tbb^d$ whenever $m > n$ and $N \geq n + m$.
\end{lemma}
\begin{proof}
  Expanding the \dvp kernel into a sum of Dirichlet kernels and applying Lemma
  \ref{thm:dirichlet_interp},
  \begin{align}
    \frac{1}{N^d} &\sum_{\tk \in \ThetaNd} p(\tk) D^d_{n,m}(\tv - \tk) \\
    &= \frac{1}{(m-n)^d} \sum_{{\ \  {\n \in \mathbb{V}^d_{n,m}}}} \frac{1}{N^d} \sum_{\tk \in \ThetaNd} p(\tk) D^d_{\n}(\tv - \tk) \\
    &= \frac{1}{(m-n)^d} \sum_{{\ \ \mathclap {\n \in \mathbb{V}^d_{n,m}}}} p(\tv) = p(\tv).
  \end{align}
\end{proof}

\subsection{Proof of \cref{thm:crest_bound}}
The upper bound of \cref{thm:crest_bound} depends on estimates of $\sum_{\tk \in \ThetaNd}
\abs{D^d_{n, m}(\tv - \tk)}$, which we collect into a pair of lemmas.
\begin{lemma}
  \label{thm:unwieldy_bound}
  Take $N \geq 2n + 1$.  Then, for all $\tv \in \Tbb^d$, 
  \begin{align}
    \sum_{\tk \in \ThetaNd}
    &\abs{D^d_{n, N-n}(\tv - \tk)}  \nonumber \\
    &\leq \left( \sup_{\tv \in \Tbb} \sum_{\tk \in \Theta_N} \abs{D_{n, N-n}(\tv - \tk)} \right)^d
      \label{eq:vp_tensor_sup_form}
    \\
    &=
      \frac{
      \displaystyle
      \left(
    \sup_{\tv \in \Tbb}
      \left\{  
    \sum_{\tk \in \Theta_N}
            \abs{
    \frac{
            \sin{\left( \frac{N \tv}{2}\right)} 
            \sin{\left(\frac{N - 2n}{2} (\tv - \tk)\right)}}
            {\sin^2{( (\tv - \tk) / 2)}}}
            \right\}
      \right)^d
      }{(N - 2n)^d}.  \nonumber
  \end{align}
\end{lemma} 
\begin{proof}
  First, we fix notation: for $\tk \in \ThetaNd$ and $k \in [N]^d$, we define $\tki = 2 \pi k_i / N$.
  Using \cref{eq:vp_kernel_factored}, we have
\begin{align}
  \sum_{\tk \in \ThetaNd}& \abs{D^d_{n, N-n}(\tv - \tk)} (N-2n)^d \nonumber\\
   = &\sum_{\tk \in \ThetaNd} \prod_{i=1}^d \abs{ \frac{ \sin{(\frac{N}{2} (\ti-\tki))}\sin{(\frac{N - 2n}{2} (\ti-\tki))}}{ \sin^2{((\ti-\tki)/2)}}}  \nonumber\\
  \leq &\left(\sup_{\tv \in \Tbb} \sum_{\tk \in \Theta_N}
         \abs{ \frac{ \sin{(\frac{N}{2}(\tv - \ti))}\sin{(\frac{N - 2n}{2} (\tv- \tk))}}{ \sin^2{((\tv- \tk)/2)}}}\right)^d
         \label{eq:vp_tensor_sup_form_intermediate}
  \\ 
  = &\left(\sup_{\tv \in \Tbb} \sum_{\tk \in \Theta_N} \abs{ \frac{ \sin{(\frac{N \tv}{2})}\sin{(\frac{N - 2n}{2} (\tv- \tk))}}{ \sin^2{((\tv- \tk)/2)}}}\right)^d, \nonumber
\end{align}
where the final step follows from $\abs{\sin{(\frac{N}{2} (\tv - 2 \pi k / N ))}} = \abs{\sin{(\frac{N
      \tv }{2})}}$ for $k \in [N]$. The bound \cref{eq:vp_tensor_sup_form} is obtained by replacing
\cref{eq:vp_tensor_sup_form_intermediate} with the definition of $D_{n, N-n}(\tv)$ given by
\cref{eq:vp_kernel_factored}.
\end{proof}

The following lemma for univariate trigonometric polynomials is key to the derivation of
\cref{eq:crest_factor_1d}. \footnote{A multivariate extension is straightforward, but not used in
  the proof of \cref{thm:crest_bound} and is omitted here.}

\begin{lemma}
  \label{thm:vp_l1_norm}
  Let $m > n$ and take $N \geq n + m$.  Then
  \begin{equation}
    \sum_{\tk \in \Theta_N} \abs{D_{n,m}(\tv - \tk)} \leq N \left(\frac{m + n}{m - n}\right)^{\frac{1}{2}}
  \end{equation}
  for all $\tv \in \Tbb$. In particular, taking $N \geq 2n+1$ and $m = N-n$ yields
  \begin{equation}
    \sum_{\tk \in \Theta_N} \abs{D_{n,N-n}(\tv - \tk)} \leq N \left(\frac{N}{N - 2n}\right)^{\frac{1}{2}}.
    \label{eq:jetter_bound_useful}
  \end{equation}
\end{lemma} 
\begin{proof}
  See \cite[Theorem 1]{Jetter2001}.
\end{proof}
We are now set to complete the proof of \cref{thm:crest_bound}.
\begin{proof}[Proof of \cref{thm:crest_bound}]
  Without loss of generality, assume $\Pn = 1$.  Then, by \cref{thm:vp_interp}, we have
  \begin{align}
    \abs{p(\tv)} &= \abs{\frac{1}{N^d} \sum_{\tk \in \ThetaNd} p(\tk) D^d_{n,N-n}(\tv - \tk)} \\
               &\leq \frac{1}{N^d} \sum_{\tk \in \ThetaNd} \abs{p(\tk) D^d_{n,N-n}(\tv - \tk)}
                 \label{eq:pf_triangle_step}
    \\
               &\leq \frac{1}{N^d} \sum_{\tk \in \ThetaNd}  \abs{D^d_{n,N-n}(\tv - \tk)}
                 \label{eq:pf_holder_step}
  \end{align} 
  where \cref{eq:pf_triangle_step} and \cref{eq:pf_holder_step} follow
  from the triangle inquality and H\"older's inequality, respectively.

  Now, applying \cref{thm:unwieldy_bound}, we have
  \begin{align}
    \abs{p(\tv)}
    &\leq N^{-d} \left( \sup_{\tv \in \Tbb} \sum_{\tk \in \Theta_N} \abs{D_{n, N-n}(\tv - \tk)} \right)^d
      \label{eq:int_vp_step_1}
    \\
         &= \frac{
           \displaystyle
           \left(
           \sup_{\tv \in \Tbb}
           \left\{  
           \sum_{\tk \in \Theta_N}
           \abs{
           \frac{
           \sin{\left( \frac{N \tv}{2}\right)} 
           \sin{\left(\frac{N - 2n}{2} (\tv - \tk)\right)}}
           {\sin^2{( (\tv - \tk) / 2)}}}
           \right\}
           \right)^d
           }{N^d (N - 2n)^d}  \label{eq:int_vp_step}, 
  \end{align}
  which implies \cref{eq:upper_bound_general}-\cref{eq:crest_bound_unwieldy}. Applying the bound
  \cref{eq:jetter_bound_useful} of \cref{thm:vp_l1_norm} to \cref{eq:int_vp_step_1} yields
  \begin{align}
    \abs{p(\tv)} \leq \left( \frac{N}{N-2n} \right)^{\frac{d}{2}} =
    \left( 1-\alpha\right)^{-\frac{d}{2}},
  \end{align}
  which establishes \cref{eq:crest_bound}.

  Finally, as $N \geq 2n+1$, we have
  $(1 - \alpha)^{-\frac{d}{2}} = 1 + \frac{d n}{N} + \mathcal{O}((dn/N)^{2})$.
 It follows that
  \begin{align*}
    \Cnd \A - \norm{p}_\infty &\leq (1 - \alpha)^{-\frac{d}{2}}\A - \norm{p}_\infty \\
                              &\leq \left(\frac{d n}{N} + \mathcal{O}(N^{-2})\right) \A \\
                              & = \mathcal{O}\left(\frac{d n}{N}\right) \norm{p}_\infty,
  \end{align*}
  where we have used $\A \leq \norm{p}_\infty$.
\end{proof}

\section{Proof of Refinement and Lower Bound For Real Trigonometric Polynomials}
\label{sec:real_proofs}
We now restrict our attention to real trigonometric polynomials. We will use the shorthand notation
$A \triangleq \max_{\tv \in \ThetaNd} p(\tv)$ and $\B \triangleq \min_{\tv \in \ThetaNd} p(\tv)$. Note both
$A$ and $B$ are (not necessarily monotonic) functions of $N$.

\subsection{Refinement}
The bound of \cref{thm:crest_bound} is at its tightest whenever $\min_{\tv \in \Tbb} p(\tv) = -
\norm{p}_\infty$ and can be loose otherwise. To see this, take $c > 0$ and consider the shifted
polynomial $\tilde{p}(\tv) = p(\tv) + c$.  Applying \cref{thm:crest_bound} yields
\begin{align}
  \norm{\tilde{p}}_\infty &\leq
                           \Cnd \norm{\tilde{p}}_{N^d, \infty} \\
                           & \leq 
                           \Cnd (\norm{p}_{N^d, \infty} + c).
                             \label{eq:bad_bound}
\end{align}
Applying the triangle inequality in advance of \cref{thm:crest_bound} results in
\begin{align}
  \norm{\tilde{p}}_\infty &\leq \norm{p}_{\infty} + c 
                            \leq 
                           \Cnd \norm{p}_{N^d, \infty} + c,
                         %   \norm{p}_\infty + c  
                         % \leq \C \left( \Pn + c \right),
\end{align}
which may be much smaller than \cref{eq:bad_bound}, but presupposes knowledge of $c$.
While we do not know this offset, it can be estimated from the samples of $\tilde{p}$.
This motivates our refined bound,
\cref{thm:crest_bound_new}, which we now prove.

\crestboundnew*
\begin{proof}[Proof of \cref{thm:crest_bound_new}]
  If $A = B$ then $p(\tv) - A$ vanishes on a set of $N^d \geq (2 n + 1)^d$ points; thus $p(\tv)$ is the
  constant polynomial $p(\tv) = A$ and \cref{eq:crest_bound_new} holds with equality.
  Define $q \in T_n^d$ as $q(\tv) \triangleq p(\tv) - \frac{A + B}{2}$, which satisfies 
  \begin{equation}
    \norm{q}_{N^d, \infty} = \abs{A - \frac{A + B}{2}} = \frac{A - B}{2}.
  \end{equation}
  Using \cref{thm:crest_bound}, 
  \begin{align}
    p(\tv) &= q(\tv) + \frac{A + B}{2} \leq \norm{q}_{\infty} + \frac{A + B}{2} \\
         &\leq \Cnd \norm{q}_{N^d, \infty}  + \frac{A + B}{2} \\
         &= \Cnd \left(\frac{A - B}{2}\right)  + \frac{A + B}{2}.
  \end{align}
\end{proof}
\cref{thm:crest_bound_new} is particularly useful for non-negative polynomials, for which
\cref{eq:crest_bound} is at its weakest. If $p$ is centered about $0$, then $A = -B$ and we recover
\cref{eq:crest_bound}. A similar shifting technique is used to establish a lower bound for real trigonometric
polynomials.

\subsection{Lower Bound}
\lowerbound*
\begin{proof}
  Define $q \in T_n^d$ as $q(\tv) = \frac{A + B}{2} - p(\tv)$, which satisfies
  \begin{equation}
    \norm{q}_{N^d, \infty} \leq \abs{ \frac{A + B}{2} - B} = \frac{A - B}{2}.
  \end{equation}
  By Theorem \ref{thm:crest_bound}, we have
\begin{align}
  \frac{A + B}{2} &= p(\tv) + q(\tv) \leq p(\tv) + \Cnd \norm{q}_{N^d, \infty} \\
  &\leq p(\tv) + \Cnd \left( \frac{A - B}{2} \right),
\end{align}
and rearranging gives \cref{eq:lower_bound}.
\end{proof}

\section{Examples}
\subsection{Univariate Example}
\cref{fig:examples} illustrates our bounds for a randomly chosen univariate real trigonometric
polynomial, $p \in
\bar{T}_8^1$, given by\footnote{The coefficients were drawn from a standard normal distribution and rounded to the first decimal point.}
\begin{equation}
  \label{eq:example_poly}
\begin{aligned}
p(\tv) &\triangleq 4.8 + 0.4 \sin(1 \tv) + 0.4 \cos(1 \tv) \\
&+ 1.0 \sin(2 \tv) + 0.1 \cos(2 \tv) 
+ 2.2 \sin(3 \tv) + 1.5 \cos(3 \tv)\\
&+ 1.9 \sin(4 \tv) + 0.8 \cos(4 \tv)
-1.0 \sin(5 \tv) + 0.1 \cos(5 \tv)\\
&+ 1.0 \sin(6 \tv) + 0.4 \cos(6 \tv)
-0.2 \sin(7 \tv) + 0.3 \cos(7 \tv)\\
&-0.1 \sin(8 \tv) + 1.5 \cos(8 \tv).
\end{aligned}
\end{equation}
Note that the bounds are not necessarily monotonic functions of $N$.  We see that an oversampling
factor of $1.3$, or $N=23$, is enough samples to certify the strict positivity of this polynomial.

\begin{figure}[t]
  \begin{subfigure}{\columnwidth}
  \centering
    \includegraphics[width=3.5in]{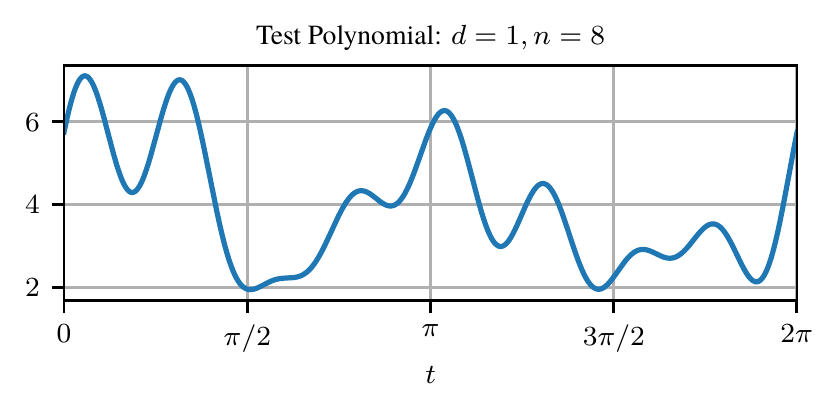}
    \caption{}
    \label{fig:example_poly} 
  \end{subfigure}
  \begin{subfigure}{\columnwidth}
    \centering
    \includegraphics[width=3.5in]{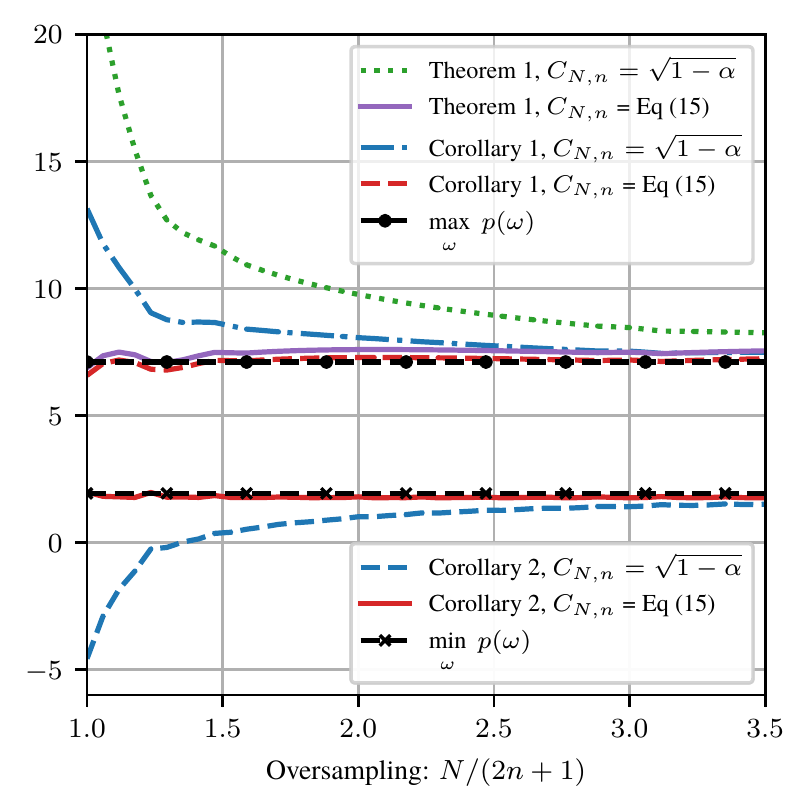}
    \caption{}
    \label{fig:example_bounds} 
\end{subfigure}
\caption{Example of upper and lower bounds for $p \in \bar{T}^1_8$ given by \cref{eq:example_poly}.
  (a): Test Polynomial. (b): Upper and lower bounds as a function of oversampling rate.}
\label{fig:examples}
\end{figure}

\subsection{Trivariate Example}
For simplicity, we take $p \in \bar{T}_n^3$ to be the Dirichlet kernel of (uniform) degree $n$; that
is, $p(\tv) = D_n^3(\tv)$  given by \cref{eq:dirichlet_kernel}.

We obtain uniform samples of $p(\tv)$ over $\ThetaNd$ by computing a zero-padded Discrete Fourier
Transform. In particular, we embed an $n \times n \times n$ array of ones into an $N \times N \times
N$ array of zeros, and apply the Fast Fourier Transform algorithm to this array. We choose $N$ to be
a favorable size for the FFT algorithm, such as a power of two. As we choose $N$ proportional to
the degree $n$ of $p$, our method scales as $\mathcal{O}(n^d \log{n})$ with $d=3$ in this example.

\cref{fig:dirichlet_example} shows the estimates obtained using
\cref{thm:crest_bound_new,thm:lower_bound} as a function of $N$ for a variety of orders $n$; the true
maximum value of $p(\tv)$ is $1$ and the minimum can be shown to be roughly $-2 / (3 \pi) \approx
-0.22$. Evaluating the bounds for $n=32$ and $N=512$ took roughly one second on a workstation with an
Intel i7-6700K CPU and $32$GB of RAM.

To draw a comparison with the sum-of-squares framework, we use the \texttt{POS3POLY} MATLAB library,
in particular the function \verb|min_poly_value_multi_general_trig_3_5| \cite{POS3POLY}. This
function finds the minimum value of a polynomial (given its coefficients) by a solving an SDP
feasibility problem using an interior point method; the maximum value is obtained by calling the
same function on $-p$. The per-iteration complexity of this method is $\mathcal{O}(n^{4d})$.

For $n = 7$, \texttt{POS3POLY} required $75$ seconds to obtain the minimum value to within
$3\times10^{-3}$; $n=8$ required $260$ seconds and found the minimum to within of $2\times10^{-3}$.
The $n=9$ case exhausted the system memory and was too large to solved on the workstation.

This is meant to be an illustrative, but certainly not exhaustive, comparison between the bounds
presented in this paper and the sum-of-squares framework. Sum-of-squares methods are especially
attractive if an exact solution is needed or if the polynomial has sparse coefficients, in which
case the complexity can be dramatically reduced.

\begin{figure}[t]
  \centering
  \includegraphics[width=3.5in]{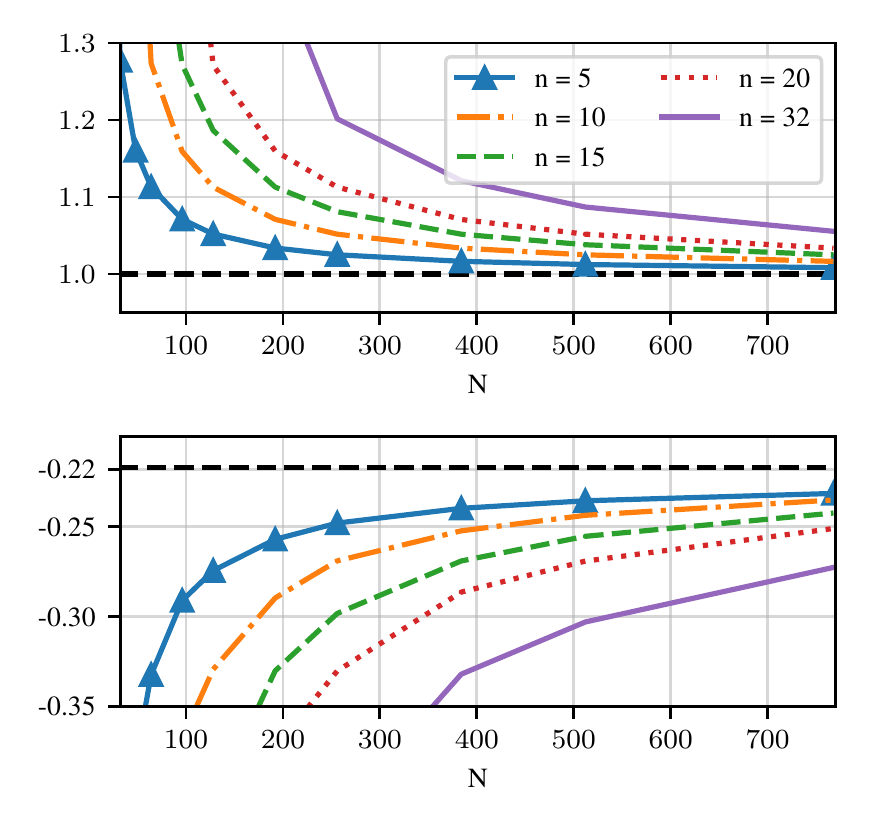}
  \caption{Upper and lower bounds for the Dirichlet kernel of 3 variables using
    \cref{thm:crest_bound_new,thm:lower_bound} as a function of number of samples, $N$.}
  \label{fig:dirichlet_example} 
\end{figure}

%%% Local Variables:
%%% mode: latex
%%% TeX-master: "trig_poly_main"
%%% End:

\section{Application to 2D Filter Bank Design}
\label{sec:filter_design}

\subsection{Perfect Reconstruction Filter Banks}

We review a few key properties of multirate perfect reconstruction filter banks before turning to our
design algorithm; see \cite{Lin1996, Do2011a} for a complete overview.

An $N_c$ channel analysis filter bank operating on $d$-dimensional signals consists of a collection
of $N_c$ \emph{analysis filters} $h_i$ and a non-singular downsampling matrix $M \in \Zbb^{d \times
  d}$. A filter bank is \emph{perfect reconstruction} (PR) if there exists a (possibly non-unique)
synthesis filter bank, consisting of a collection of $N_c$ \emph{synthesis filters}, $g_i$,
and the upsampling matrix $M$, that reconstructs a signal from its analyzed version. An analysis
filter bank, along with its corresponding synthesis filter bank, are illustrated in
\cref{fig:filterbank}. If the filter bank is PR then $\hat{x} = x$. In what follows,
a 'filter bank' indicates an analysis filter bank unless otherwise specified.

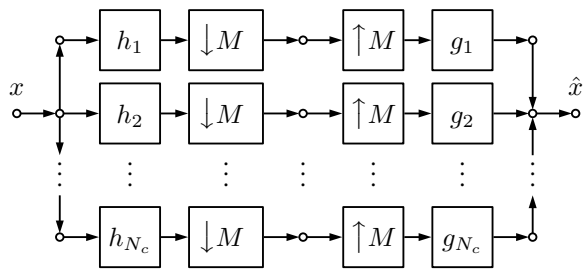
\begin{figure}[t]
  \centering
  \begin{tikzpicture}[decoration={brace,mirror, amplitude=7}]
%% Robust vdots that won't skip space at top
\makeatletter
\DeclareRobustCommand{\rvdots}{%
  \vbox{
    \baselineskip4\p@\lineskiplimit\z@
    \kern-\p@
    \hbox{.}\hbox{.}\hbox{.}
  }}

\makeatother
\matrix (m1) [row sep=1.5mm, column sep=1.8mm]
    {
    %--------------------------------------------------------------------
    % Top row
    %--------------------------------------------------------------------
      \node[coordinate]         (m00) {}; &
      \node[coordinate]         (m01) {}; &
      \node[dspnodeopen]        (m02) {}; &
      \node[coordinate]         (m03) {}; &
      \node[dspsquare]          (m04) {$h_1$}; &
      \node[coordinate]         (m05) {}; &
      \node[dspsquare]          (m06) {\ $\downsamplertext{M}$ \ }; &
      \node[coordinate]         (m07) {}; &
      \node[dspnodeopen]        (m08) {};  &
      \node[coordinate]         (m09) {}; &
      \node[dspsquare]          (m010) {$\upsamplertext{M}$}; &
      \node[coordinate]         (m011) {}; &
      \node[dspsquare    ]      (m012) {$g_1$}; &
      \node[coordinate]         (m013) {}; &
      \node[dspnodeopen]        (m014) {}; &
\\

    %--------------------------------------------------------------------
    % Base row with input
    %--------------------------------------------------------------------
      \node[dspnodeopen]        (m10) {$x$};  &
      \node[coordinate]         (m11) {}; &
      \node[dspnodeopen]        (m12) {}; &
      \node[coordinate]         (m13) {}; &
      \node[dspsquare]          (m14) {$h_{2}$}; &
      \node[coordinate]         (m15) {}; &
      \node[dspsquare]          (m16) {\ $\downsamplertext{M}$ \ }; &
      \node[coordinate]         (m17) {}; &
      \node[dspnodeopen]        (m18) {};  &
      \node[coordinate]         (m19) {}; &
      \node[dspsquare]          (m110) {$\upsamplertext{M}$}; &
      \node[coordinate]         (m111) {}; &
      \node[dspsquare]          (m112) {$g_{2}$}; &
      \node[coordinate]         (m113) {}; &
      \node[dspnodeopen]        (m114) {}; &
      \node[coordinate]         (m115) {}; &
      \node[dspnodeopen]        (m116) {$\hat{x}$}; &
\\
    %--------------------------------------------------------------------
    % Dots row
    %--------------------------------------------------------------------
      \node[coordinate]         (m20) {}; &
      \node[coordinate]         (m21) {}; &
      \node                     (m22) {\rvdots}; &
      \node[coordinate]         (m23) {}; &
      \node                     (m24) {\rvdots}; &
      \node[coordinate]         (m25) {}; &
      \node                     (m26) {\rvdots}; &
      \node[coordinate]         (m27) {}; &
      \node                     (m28) {\rvdots};  &
      \node[coordinate]         (m29) {}; &
      \node                     (m210) {\rvdots};  &
      \node[coordinate]         (m211) {}; &
      \node                     (m212) {\rvdots};  &
      \node[coordinate]         (m213) {}; &
      \node                     (m214) {\rvdots}; &
      \node[coordinate]         (m215) {}; &
\\

    %--------------------------------------------------------------------
    % Bottom row
    %--------------------------------------------------------------------
      \node[coordinate]         (m30) {}; &
      \node[coordinate]         (m31) {}; &
      \node[dspnodeopen]        (m32) {}; &
      \node[coordinate]         (m33) {}; &
      \node[dspsquare]          (m34) {$h_{N_c}$}; &
      \node[coordinate]         (m35) {}; &
      \node[dspsquare]          (m36) {\ $\downsamplertext{M}$ \ }; &
      \node[coordinate]         (m37) {}; &
      \node[dspnodeopen]        (m38) {};  &
      \node[coordinate]         (m39) {}; &
      \node[dspsquare]          (m310) {$\upsamplertext{M}$}; &
      \node[coordinate]         (m311) {}; &
      \node[dspsquare]          (m312) {$g_{N_c}$}; &
      \node[coordinate]         (m313) {}; &
      \node[dspnodeopen]        (m314) {}; &
      \node[coordinate]         (m315) {}; &
      \node[coordinate]         (m316) {}; &
      \\
    };

\foreach \i [evaluate = \i as \j using int(\i+2)] in {2,4,...,12}
		\draw[dspconn] (m0\i) -- (m0\j);
\foreach \i [evaluate = \i as \j using int(\i+2)] in {2,4,...,14}
		\draw[dspconn] (m1\i) -- (m1\j);
\foreach \i [evaluate = \i as \j using int(\i+2)] in {2,4,...,12}
		\draw[dspconn] (m3\i) -- (m3\j);

 	\draw[dspconn] (m10) -- (m12);
 	\draw[dspconn] (m12) -- (m02);
 	\draw[dspconn] (m12) -- (m22);
 	\draw[dspconn] (m22) -- (m32);

 	\draw[dspconn] (m014) -- (m114);
 	\draw[dspconn] (m214) -- (m114);
 	\draw[dspconn] (m314) -- (m214);

     % \draw [decorate] ([yshift=-5mm]m32.west) --node[below=2mm]{$\mathcal{H}$} ([yshift=-5mm] m37.south);
     % \draw [decorate] ([yshift=-5mm]m39.west) --node[below=2mm]{$\mathcal{H^\dagger}$} ([yshift=-5mm] m316.east);

\end{tikzpicture}

%%% Local Variables:
%%% mode: latex
%%% TeX-master: "trig_poly_main"
%%% End:
  \caption{An $N_c$ channel multi-rate filter bank with analysis filters $h_i$ and synthesis filters $g_i$.}
  \label{fig:filterbank}
\end{figure}

We consider finite impulse response (FIR) filters, and for simplicity, we restrict our attention to
impulse responses with a square support. A real (square) $d$-variate (or $d$-dimensional) FIR filter
$h$ of length $n$ is a function $h : \Zbb^d \to \Rbb$ such that $h[m] = 0$ if $m_i < 0$ or $m_i \geq
n$ for any $0 \leq i < d$.

A multidimensional discrete-time signal is a function $x: \Zbb^d \to \Rbb$.
Downsampling a signal $x$ by a non-singular integer matrix $M$ retains only the samples on the lattice
generated by $M$; that is, integer vectors of the form $v = Mt$. The simplest choice of downsampling
matrix is $M = s I_d$, where the integer $s \geq 1$ controls the downsampling factor and $I_d$ is
the identity matrix in $d$ dimensions. We will refer to this as the \emph{uniform} downsampling
scheme.

The $i$-th polyphase component of a signal $x$ is a function $\hat{x}^i : \Zbb^d \to \Rbb$ obtained
by shifting and downsampling $x$. In particular, $\hat{x}^i[m] = x[Mm + v_i]$ for $m \in \Zbb^d$,
where $v_i$ is an integer vector of the form $Mt$ and $t \in [0, 1)^d$. There are $\abs{M}
\triangleq \det{M}$ such integer vectors, and each generates one polyphase component of the
signal. The $z$-transform of the $i$-th polyphase component of $x$ is $\hat{X}^{i}(z) = \sum_{n \in
  \Zbb^d} x[Mn + v_i] z^{-n}$, where $z \in \Cbb^d$ and $z^{-n} = z_1^{-n_1} z_2^{-n_2}\hdots
z_d^{-n_d}$.

The polyphase decomposition of an analysis filter is defined in a similar fashion. The $i$-th
polyphase component of the analysis filter $h$ is $\hat{h}^{k}[m] = h[Mm - v_i]$;  note the
difference in sign when compared to the definition of $\hat{x}^i$.

A $d$-dimensional filter bank with filters $\left\{h_i\right\}_{i=1}^{N_c}$ and downsampling matrix
$M$ has a polyphase matrix $\hat{\mathbf{H}}(z) \in \Cbb^{N_c \times \abs{M}}$ formed by stacking
the polyphase components of each analysis filter into a row vector, and stacking the $N_c$ rows into a
matrix. Explicitly,
\begin{equation}
  \label{eq:polyphase_mat}
  \hat{\mathbf{H}}(z) \triangleq \begin{bmatrix}
    \hat{H}_{0}^{0}(z) &  \hat{H}_{0}^{1}(z) & \hdots & \hat{H}_{0}^{\abs{M}-1}(z) \\
    \hat{H}_{1}^{0}(z) &  \hat{H}_{1}^{1}(z) & \hdots & \hat{H}_{1}^{\abs{M}-1}(z) \\
    \vdots & \vdots & \ddots & \vdots \\
    \hat{H}_{N^c-1}^{0}(z) &  \hat{H}_{N^c-1}^{1}(z) & \hdots & \hat{H}_{N_c-1}^{\abs{M}-1}(z) \\
  \end{bmatrix},
\end{equation}
where $\hat{H}^{k}_{i} (z)$ is the $z$-transform of the $k$-th polyphase component of the $i$-th
filter. The entries of $\hat{\mathbf{H}}(z)$ are multi-variate Laurent polynomials in $z \in \Cbb^d$
and become trigonometric polynomials when restricted to the unit circle; that is, $z = e^{j \omega}$
with $\omega \in \Tbb^d$. In a customary abuse of notation, we write $\hat{\mathbf{H}}(\omega)
\triangleq \hat{\mathbf{H}}(e^{j \omega})$.

There are deep connections between perfect reconstruction filter banks and redundant signal
expansions using \emph{frames} \cite{Cvetkovic1998, Bolcskei1998, Christensen2003, Strang1996}. In
particular, oversampled perfect reconstruction filter banks implement an \emph{frame expansion}.
Associated with a perfect reconstruction filter bank are a pair of scalars, the \emph{upper and
  lower frame bounds}, defined by
\begin{align}
  A &\triangleq \mathrm{ess \ sup}_{\omega \in \Tbb^d, m=1,\hdots \abs{M}} \ \lambda_n(\omega), \\
  B &\triangleq \mathrm{ess \ inf}_{\omega \in \Tbb^d, m=1,\hdots \abs{M}} \ \lambda_n(\omega) 
\end{align}
where $\lambda_n(\omega)$ is an eigenvalue of the matrix $\hat{\mathbf{H}^*}(\omega) \hat{\mathbf{H}}(\omega)$.
If $A = B$ the frame is said to be \emph{tight}. The ratio $A / B$ is the \emph{frame
  condition number}; if $A /B \approx 1$, the frame is said to be \emph{well-conditioned}.
The frame bounds of a filter bank determine important numerical properties such
as sensitivity to perturbations, and the frame condition number serves a similar role as the
condition number of a matrix. 

The synthesis filter bank also admits a polyphase decomposition. The $i$-th polyphase component of a
synthesis filter $g$ is $\hat{g}^{k}[m] = g[Mm + v_i]$. The synthesis polyphase matrix is of size
$\abs{M} \times N_c$ and has entries
\begin{equation}
  \label{eq:polyphase_synth_mat}
  \hat{\mathbf{G}}(z) \triangleq \begin{bmatrix}
    \hat{G}_{0}^{0}(z) &  \hat{G}_{1}^{0}(z) & \hdots & \hat{G}_{\abs{M}-1}^{0}(z) \\
    \hat{G}_{0}^{1}(z) &  \hat{G}_{1}^{1}(z) & \hdots & \hat{G}_{\abs{M}-1}^{1}(z) \\
    \vdots & \vdots & \ddots & \vdots \\
    \hat{G}_{0}^{N^c-1}(z) &  \hat{G}_{1}^{N^c-1}(z) & \hdots & \hat{G}_{\abs{M}-1}^{N_c-1}(z) \\
  \end{bmatrix}.
\end{equation}

If a pair of analysis and synthesis filter banks share the PR property, then $\hat{\mathbf{G}}(z)
\hat{\mathbf{H}}(z) = I_{\abs{M}}$, where $I_{\abs{M}}$ is the $\abs{M} \times \abs{M}$ identity
matrix. That is, $\hat{\mathbf{G}}(z)$ is a left inverse for $\hat{\mathbf{H}}(z)$. If $N_c >
\abs{M}$, the filter bank is said to be \emph{oversampled}, and the synthesis filter bank is not
unique.  A particular choice is the \emph{minimum-norm synthesis filter bank}, given by
\begin{equation}
  \label{eq:min_norm_synthesis_fb}
  \hat{\mathbf{H}}^\dagger(z) \triangleq \left( \tilde{\mathbf{H}}(z) \hat{\mathbf{H}}(z) \right)^{-1} \tilde{\mathbf{H}}(z),
\end{equation}
where the \emph{para-conjugate} matrix $\tilde{\mathbf{H}}(z)$ is obtained by conjugating the
polynomial coefficients of $\hat{\mathbf{H}}(z)$, replacing the argument $z$ by $z^{-1}$, and
transposing the matrix. On the unit circle, $\hat{\mathbf{H}}^\dagger(\omega) = \left(
  \hat{\mathbf{H}}^*(\omega) \hat{\mathbf{H}}(\omega) \right)^{-1} \hat{\mathbf{H}}^*(\omega)$.

A filter bank is perfect reconstruction if and only if its polyphase matrix has full column rank on
the unit circle \cite{Vetterli1987, Cvetkovic1998}. As the matrix $\hat{\mathbf{H}}^*(\omega)
\hat{\mathbf{H}}(\omega)$ is positive semidefinite, the perfect reconstruction property holds if and
only if the 
trigonometric polynomial
\begin{equation}
p_H(\omega) \triangleq \det \left( \hat{\mathbf{H}}^*(\omega) \hat{\mathbf{H}}(\omega)\right)
\end{equation}
 is strictly positive.  This property is key to our proposed filter bank design algorithm.

 The degree of $p_H(\omega)$ depends on the filter length and the downsampling matrix. To
 illustrate, we bound from above the degree of $p_H(\omega)$ when using separable downsampling. After
 downsampling by $M = s I_d$, a FIR filter of length $n$ retains at most $\ceil{n / s}$ entries
 along each dimension; thus the polyphase component $\hat{H}_i^k(\omega)$ has maximum component
 degree $n^\prime \triangleq \ceil{n/s} - 1$. Note that $\hat{H}_i^k(\omega)$ contains only negative
 powers of $\omega$; that is, $\hat{H}_i^k(\omega)\in\spanset{e^{-j k \cdot \omega} : \omega \in
   \Tbb^d, k \in \Zbb^d, \ n^\prime \leq k_i \leq 0}$. As such, the trigonometric polynomials
 $(\hat{H}^k_i(\omega))^* \hat{H}^l_i(\omega)$ remain in $T_{n^\prime}^d$ and the entries of the
 matrix $\hat{\mathbf{H}}^*(\omega) \hat{\mathbf{H}}(\omega)$ are in the same space.

At worst, the determinant includes the product of $\abs{M} =
s^d$ polynomials of degree $n^\prime$, and so $p_H \in \bar{T}^d_{m}$ with 
\begin{equation}
  \label{eq:ph_deg_bound}
  m \leq s^d (\ceil{n / s} -1). 
\end{equation}
Taking $n=12, s=2$ and $d=2$, we have $p_H \in T^2_{20}$.

\subsection{Filter Bank Design: Analysis}
The simplest multi-dimensional PR filter banks apply a 1D PR filter bank independently to each
signal dimension; for example, in 2D, to the horizontal and vertical directions. These
\emph{separable} filters are
written as a product of multiple 1D filters and suffer from limited directional sensitivity. The
design and construction of \emph{non-separable} multi-dimensional filter banks is difficult due to the lack
of a spectral factorization theorem \cite{Venkataraman1994}; indeed, directly verifying the perfect
reconstruction condition for a 2D filter bank is equivalent to determining the minimum value of a
trigonometric polynomial and is thus NP-Hard \cite{Murty1987, Parrilo2003}.

Some 2D PR filter banks, such as curvelets, have been hand-designed \cite{Candes2004, Candes2006}.
Other design methods include variable transformations applied to a 1D PR filter bank \cite{Lin1996,
  Do2011a}, modulating a prototype filter \cite{Lin1996}, invoking tools from
algebraic geometry \cite{Zhou2005}, or by solving an optimization problem \cite{Lu1998, Chen2007}.

Optimizing a filter bank subject to the PR condition is a semi-infinite optimization problem: we
have a finite number of design  variables, namely the filter coefficients, and the resulting
polyphase matrix must be positive semidefinite over $\Tbb^d$.  

One approach is to carefully parameterize the filter bank architecture in such a way that guarantees
the PR property \cite{Venkataraman1994, Chen2007}. A different approach is to relax the PR condition
to \emph{near} PR, and minimize the resulting reconstruction error using an iterative
algorithm \cite{Lu1998}.

We use a different approach: we relax the semi-infinite problem into a finite one, then use
\cref{thm:lower_bound_condition} to certify that the solution of the relaxed problem is also a
solution to the original problem. In particular, we design the filter bank such that $p_H(\omega)$
is strictly positive over the finite collection of sampling points $\Theta_N^d$.
\cref{thm:lower_bound_condition} tells us that if the bounds \cref{eq:lower_bound_cond_cnd} or
\cref{eq:lower_bound_cond_alpha} are satisfied, then $p_H(\omega)$ is strictly positive over all of
$\Tbb^d$, and the filter bank is thus PR.

We design our filter banks with an eye towards the bounds of \cref{thm:lower_bound_condition}: we
want the maximum and minimum sampled values of $p_H(\omega)$ to be close to one another, so that the
bounds \cref{eq:lower_bound_cond_cnd,eq:lower_bound_cond_alpha} are satisfied for smaller values of
$N$.

Our filter design approach is highly flexible. It applies to arbitrary filter lengths, any
non-singular decimation matrix, and will design PR filter banks in any number of dimensions. For
simplicity we focus on designing real, 2D filter banks ($d=2$) but our approach can be modified for
$d$-dimensional complex filters.

We begin by specifying the number of channels, $N_c$, downsampling matrix $M$, and filter size. We
require that $N_c \geq \abs{M}$ so that the PR condition can hold.  For simplicity, we use
downsampling of the form $M = sI_2$, but our method can design filter banks using non-separable
(\eg, quincunx) downsampling matrices. We also constrain each filter to be of size $n \times n$,
although this can be easily relaxed.

With these parameters set, we calculate the maximum degree $m$ of $p_H(\omega)$ using
\cref{eq:ph_deg_bound}. Next, we select the number of sampling points, $N$, to use during the design
process.  The conditions of \cref{thm:lower_bound_condition} require we take $N \geq 2m + 1$,  but in
practice we take $N > 4m$ so that we can tolerate larger values of $\kapt$ while still certifying
the perfect reconstruction property.

The $i$-th $n \times n$ filter will be written $h_i$, and we group the filters into a tensor $H \in
\Rbb^{N_c \times n \times n}$. The Discrete-Time Fourier Transform of the $i$-th filter is
\begin{equation}
  h_i(\omega) = \sum_{m \in [n]^2} h_i[m] e^{j\omega\cdot m} \quad \omega \in \Tbb^2,
\end{equation}
and the squared magnitude response of $h_i$ is $\abs{h_i(\omega)}^2$.

Our goal is to design a perfect reconstruction filter bank where the magnitude response of the
$i$-th channel matches a desired real and non-negative magnitude response $D_i(\omega)$ for $\omega
\in \Tbb^2$. We use a weighted quadratic penalty that measures the discrepancy between the magnitude
response of a candidate filter and the $D_i$ at the 2D-DFT samples $\ThetaNt$.
Our filter design function is written
\begin{equation}
  f(H, D) \triangleq \sum_{i=1}^{N_c} \sum_{\omega \in \ThetaNt} W_i(\omega)\cdot \abs{  \abs{\hat{h}_i(\omega)}^2 - D_i(\omega)}^2,
\end{equation}
where we have introduced weighting functions $W_i(\omega)$ to control the importance given to the
passband, transition band, and stop band.   If $D_i$ is not specified for some $i$, we take
$W_i(\omega)$ to be uniformly zero;  then $h_i$ does not contribute to $f(H, D)$ but may contribute
to the PR property of the filter bank.  

We emphasize that other choices of a design function are possible; for instance, one could use a
minimax criterion and minimize the maximum deviation between $\hat{h}_i(\omega)$ and $D_i(\omega)$.
Elsewhere, we have used a similar approach to learn signal-adapted undecimated perfect
reconstruction (analysis) filter banks under a sparsity-inducing criterion \cite{Pfister2018a}.

In some cases, the filter design function alone may promote perfect reconstruction filter banks- for
instance, when designing a non-decimated $(M = I_d)$ filter bank where the desired magnitude
responses satisfy a partition-of-unity condition. In general, though, this term is not enough. We
add an additional regularization term to encourage filter banks that can be certified as perfect
reconstruction using \cref{thm:lower_bound_condition}. Our regularizer is given by
\begin{equation}
  \label{eq:regularizer}
  R(H) \triangleq \alpha \sum_{i=1}^{N_c} \norm{h_i}_F^2 + \sum_{\omega \in \ThetaNt} \beta p_H(\omega)^2 - \gamma \log{p_H(\omega)},
\end{equation}
where the non-negative scalars $\alpha, \beta, \gamma$ are tuning parameters. The first term
prohibits the filter norms from becoming too large. The second and third terms apply the function
$\omega \mapsto p_H(\omega)^2 - \log{p_H(\omega)}$ for each $\omega \in \ThetaNt$. The negative
logarithm barrier function becomes large when $p_H(\omega)$ goes to zero and the quadratic part
discourages large values of $p_H(\omega)$.

Together, these terms ensure the matrix $\hat{\mathbf{H}}(\omega)$ is left invertible and
well-conditioned for each $\omega \in \ThetaNt$. They also ensure $p_H(\omega)$ does not grow too
large over the sampling set. These properties ensure $p_H(\omega)$ is strictly positive and doesn't
vary too much over $\ThetaNt$; thus, by \cref{thm:lower_bound_condition}, $R(H)$ promotes
well-conditioned perfect reconstruction filter banks. We emphasize that this regularizer, as well as
the filter design function, are only computed over
on the discrete set $\ThetaNt$; passage to the continuous case is handled by
\cref{thm:lower_bound_condition}.

Our designed filter bank is the solution to the optimization problem 
\begin{equation}
  \label{eq:optimization}
  \min_{H \in \mathcal{C}} f(H, D) + R(H),
\end{equation}
where the constraint set $\mathcal{C}$ reflects any additional constraints on the filters, \eg
symmetry.  

This minimization can be solved using standard first order methods such as gradient
descent. The main challenge is calculating the gradient of $\log(p_H(\omega))$, which is
unwieldy for all but the shortest filters. A finite-difference approximation to the gradient can
suffice, but we have had success using the reverse-mode automatic differentiation capabilities of
the \texttt{autograd}\footnote{\url{https://github.com/HIPS/autograd}} and
\texttt{Pytorch}\footnote{\url{http://pytorch.org/}} Python packages. Our algorithm is implemented
in \texttt{Pytorch} and runs on an NVidia Titan X GPU.

\subsection{Experiment: Design of a curvelet-like filter bank}
\label{sub:curvelet_analysis}
Our goal is to design a filter bank that approximates the discrete curvelet filter bank. Our desired
magnitude responses are obtained from the frequency space tiling illustrated in \cref{fig:tiling};
each channel should have a pass-band corresponding to a cell in this tiling. As the magnitude
frequency response of a real filter is symmetric, \eg $\abs{\hat{h}(\omega_1, \omega_2)} =
\abs{\hat{h}(-\omega_1, -\omega_2)}$, $17$ filters are needed for the desired partitioning.
We use uniform downsampling by a factor of $2$, that is, $M = 2I_2$.  The filter bank is roughly
$4\times$ oversampled.

The weighting functions $W_i(\omega)$ were set to $1$. We set $\beta=10$ and $\alpha = \gamma = 1$.
We used $5000$ iterations of the Adam optimization algorithm with a learning rate of $10^{-2}$
\cite{Kingma2014}. The optimization completed in under one minute for all tasks.

We designed two filter banks; one with $8\times 8$ filters and the other with $11 \times 11$
filters. We used $N = 64$ for both cases.
The final filter banks and their magnitude responses are
shown in \cref{fig:designed_filterbank}.

We tested two methods to initialize the algorithm.  In the first method, we take an $N \times
N$ inverse DFT of the desired magnitude response, $D_i$, and extract the $n \times n$ central region
of the resulting impulse response.  Our second method is a simple random initialization.  Both
methods perform equally well in our design task.

We use \cref{thm:lower_bound_condition} to verify the final filter banks are perfect reconstruction.
For our filter bank with $8 \times 8$ filters, the bound \cref{eq:ph_deg_bound} indicates $p_H \in
\bar{T}^2_{12}$. Our sufficient condition in \cref{thm:lower_bound_condition} for strict positivity
requires $\kappa_{64} \leq 4.4$, with $\kappa_N$ given by \cref{eq:lower_bound_cond_cnd}. We
computed $p_H(\omega)$ over all points in $\Theta_{64}^2$, and used these values to compute
$\kappa_{64}$. We found $\kappa_{64} = 1.3$ for the designed filter bank, and thus the filter bank
is perfect reconstruction. When using $11 \times 11$ filters, we have $p_H \in \bar{T}^2_{20}$. This
filter bank too is perfect reconstruction, as $\kappa_{64} = 1.8 \leq 2.2$.

\begin{figure}
  \centering
  \includegraphics[scale=0.3]{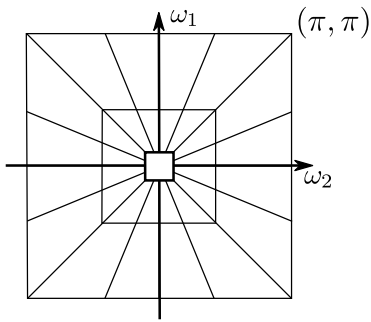}
  \caption{Desired tiling of frequency space.}
  \label{fig:tiling}
\end{figure}

\begin{figure}
  \centering
  \begin{subfigure}{0.4\columnwidth}
    \centering
    \includegraphics[scale=0.2,angle=90]{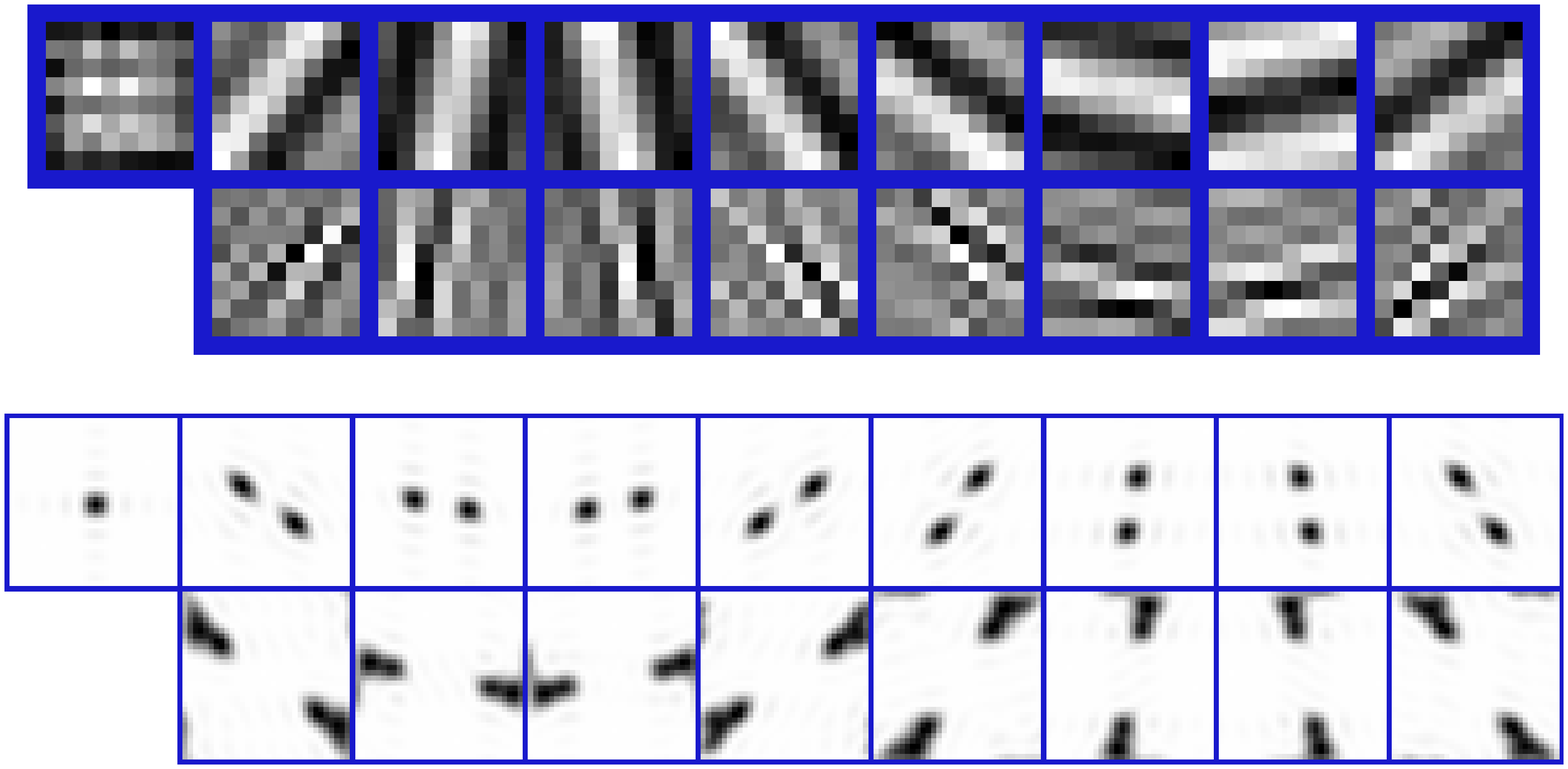}
    \caption{}
  \end{subfigure}
  \begin{subfigure}{0.4\columnwidth}
    \centering
    \includegraphics[scale=0.2,angle=90]{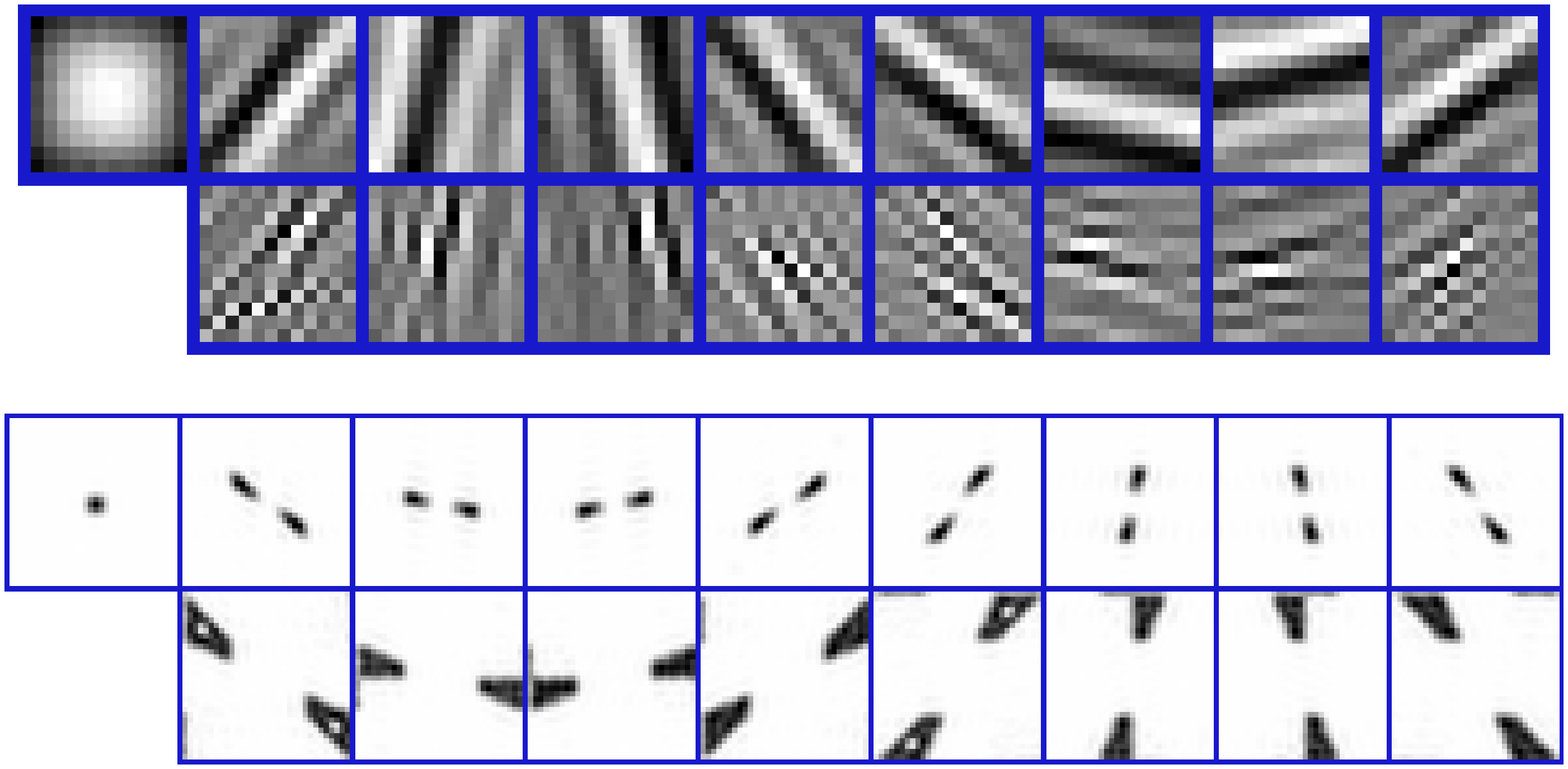}
    \caption{}
  \end{subfigure}
  \caption{Optimized 17 channel filter bank.  The left column of each subfigure shows the filter impulse
    response. The right column shows the magnitude frequency response, with $\omega = 0$ located at
    the center of each blue box.
    (a) 17 channel filter bank with $8\times 8$ filters.
    (b) 17 channel filter bank with $11 \times 11$ filters.}
  \label{fig:designed_filterbank}
\end{figure}

\subsection{Filter Bank Design: Synthesis}
Our filter design problem has focused exclusively on the analysis portion of the filter bank, but in
many applications the synthesis filter bank is equally important.

We focus on the oversampled case, \ie $N_c > \abs{M}$. The choice of synthesis filter bank is not
unique. We have already seen one possible choice- the minimum-norm synthesis filter bank
\cref{eq:min_norm_synthesis_fb}, which can be obtained explicitly once the analysis filter bank has
been designed. In general, the minimum-norm synthesis filter bank consists of infinite impulse
response (IIR) filters \cite{Bolcskei1999,Cvetkovic1998}.

In many applications, IIR filters are not practical-  only FIR filters can be used, and short FIR
filters are especially desirable from a computational perspective.   

Fortunately, the redundancy of an oversampled filter bank affords us design flexibility. 
Sharif investigated when a generic\footnote{A ``generic'' filter bank is one that is drawn at
  random; \ie not a pathological choice.} one-dimensional oversampled PR analysis filter bank
admits a synthesis filter bank with short FIR filters. He found that almost all sufficiently
oversampled PR analysis filter banks have such a synthesis filter bank, and obtained bounds on the
minimum synthesis filter length \cite{Sharif2011}. The bounds depend only on the number of channels,
downsampling factor, and analysis filter length, but not on the filter coefficients themselves.

We have a few options if a FIR synthesis filter bank is desired.
The simplest solution is to truncate the (IIR) minimum-norm synthesis filters to a particular
length. Indeed, a well-conditioned PR analysis filter bank has minimum-norm synthesis filters with
coefficients that exhibit decay exponentially with filter length, implying that the minimum-norm
synthesis filter bank can be well-approximated by FIR filters \cite{Strohmer2001}.  

A second option is to use tools from algebraic geometry to find an FIR synthesis filter bank, if one
exists \cite{Zhou2005a}.  

We adopt a third option: we incorporate the desire for an FIR synthesis filter bank directly into the
design problem. We add an additional set of FIR filters, denoted $\left\{g_i\right\}_{i=1}^{N_c}$,
to the design parameters. The synthesis filters need not be the same length as the analysis filters.
Our goal is for the polyphase matrix associated with the synthesis filter bank,
$\hat{\mathbf{G}}(\omega)$, to be a left inverse of the analysis polyphase matrix.  This condition
is represented by the constraint
\begin{equation}
\hat{\mathbf{G}}(\omega) \hat{\mathbf{H}}(\omega) = I_{\abs{M}}.
\label{eq:fir_li_constraint}
\end{equation}
In practice, we solve an unconstrained problem using the
quadratic penalty method: we penalize the distance between 
$\hat{\mathbf{G}}(\omega) \hat{\mathbf{H}}(\omega)$ and $I_{\abs{M}}$
for each $\omega \in \Theta_N^2$ using the Frobenius norm \cite{Wright2006}.  Our modified design problem
is given by  
\begin{equation}
  \label{eq:optimization_2}
  \min_{H, G \in \mathcal{C}} f(H, D) + R(H) + \lambda \sum_{\omega \in \Theta_N^2}
  \norm{\hat{\mathbf{G}}(\omega)\hat{\mathbf{H}}(\omega) - I_{\abs{M}}}_F^2.
\end{equation}
We again use a first order method, but now increase $\lambda$ as a function of the iteration number
so as to ensure $\hat{\mathbf{G}}(\omega)$ is a left inverse of $\hat{\mathbf{H}}(\omega)$.

As before, our new regularizer is evaluated only over $\Theta_N^2$, not $\Tbb^2$. For fixed, finite
filter lengths, the entries of $\hat{\mathbf{G}}(\omega)\hat{\mathbf{H}}(\omega)$ are real
trigonometric polynomials of bounded degree, and we can use the bounds of
\cref{thm:crest_bound_new,thm:lower_bound} to either ensure the constraint
\cref{eq:fir_li_constraint} holds over $\Tbb^2$ or to estimate and bound the amount that the
constraint has been violated.

\subsection{Experiment: Filter Bank Design with FIR Synthesis Filters}
\label{sub:fir_synth}

We repeat the design experiment from \cref{sub:curvelet_analysis} using the new objective function
\cref{eq:optimization_2}. As before, we use 17 channels and take $M = 2 I_2$, leading to a roughly
$4\times$ oversampled filter bank.  We work with $11 \times 11$ filters.  We used $5000$ iterations
of the Adam optimization algorithm with a learning rate of $10^{-2}$, and set the 
parameter $\lambda := \log_2(i)$ at iteration $i$.

\cref{fig:inverse_filter_comparison} collects the design results.
\cref{fig:inverse_filter_comparison}a
shows the $11 \times 11$ analysis filters embedded into a larger $40 \times 40$ region.  This is done to
facilitate comparison with the minimum-norm synthesis filters, shown in
\cref{fig:inverse_filter_comparison}b.  The minimum-norm synthesis filters exhibit fast decay, as
expected for a well-conditioned filter bank.  The designed FIR synthesis filters, the $\left\{ g_i
\right\}_{i=1}^{N_c}$, are shown in \cref{fig:inverse_filter_comparison}c.  These filters have no
discernible structure.  However, we computed 
$\norm{\hat{\mathbf{G}}(\omega)\hat{\mathbf{H}}(\omega) - I_{\abs{M}}}_F^2 < 10^{-7}$ for each
$\omega \in \Theta_{128}^2$, 
this is a synthesis filter bank for $\hat{\mathbf{H}}$.   Indeed, passing the standard
\texttt{barbara} test image through the pair of analysis and synthesis filter banks yielded a
reconstruction peak signal to noise ratio (PSNR) of more than $80$ dB.

\cref{fig:inv_filter_decay} illustrates the coefficient decay properties of the minimum-norm
synthesis filters.  We show the square root of the absolute value of the filter coefficients to
compress the dynamic range of the image.  We see the expected exponential decay of filter
coefficients associated with a well-conditioned filter bank\cite{Strohmer2001}.

\begin{figure}[ht]
  \centering
  \includegraphics{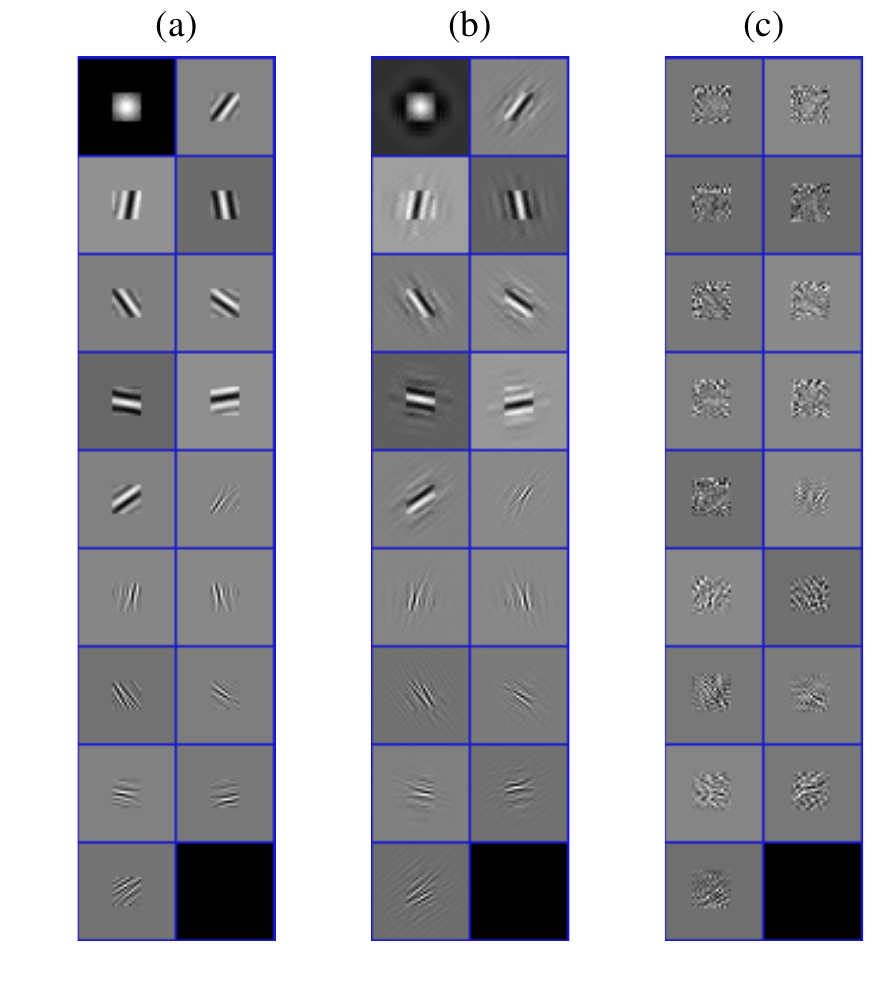}
  \caption{Analysis and Synthesis filters for filter bank designed in \cref{sub:fir_synth}. (a)
    Designed $11 \times 11$ analysis filters embedded into $40 \times 40$ filter. (b) Minimum-norm
    synthesis filters, obtained using \cref{eq:min_norm_synthesis_fb}. The filters exhibit fast
    coefficient decay; see \cref{fig:inv_filter_decay}. (c) Designed $16 \times 16$ FIR synthesis filters.
  }
  \label{fig:inverse_filter_comparison} 
\end{figure}

\begin{figure}[ht]
  \centering
  \includegraphics[trim={3cm 0 0 0},clip]{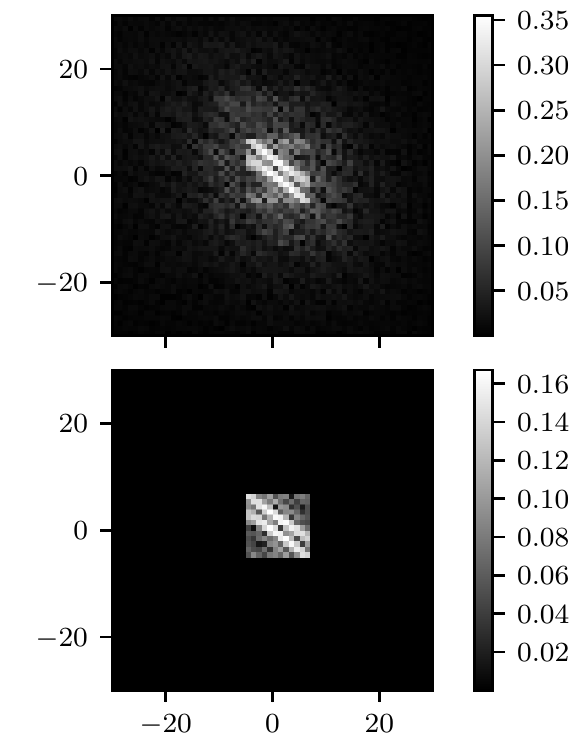}
  \caption{Square-root of absolute value of filter coefficients from one of the filters in
  \cref{fig:inverse_filter_comparison}a.  Top:  Minimum-norm synthesis filter exhibits fast
    coefficient decay, can be approximated with FIR filter.  Bottom:  FIR analysis filter.}
  \label{fig:inv_filter_decay} 
\end{figure}

% \begin{figure}
%   \centering
%   \begin{subfigure}{0.4\columnwidth}
%     \includegraphics[scale=0.2,angle=90]{fig/lp_no_ds_15_tap}
%     \caption{}
%   \end{subfigure}
%   \begin{subfigure}{0.4\columnwidth}
%     \includegraphics[scale=0.2,angle=90]{fig/lp_2_ds_11_tap}
%     \caption{}
%   \end{subfigure}
%   \caption{Optimized filter bank with symmetric filters. The left column of
%     each subfigure shows the filter impulse response. The right column shows the magnitude frequency
%     response. DC is at the center of each blue box. (a) 10 undecimated channel filter bank with $15\times 15$
%     filters. (b) 10 channel filter bank with $11 \times 11$ filters and uniform downsampling by 2.}
% \end{figure}

%%% Local Variables:
%%% mode: latex
%%% TeX-master: "trig_poly_main"
%%% End:

\section{Conclusion}
\label{sec:conclusion}

We have proposed a fast and simple method to estimate the extremal values of a multivariate
trigonometric polynomial directly from its samples. We have extended an existing upper bound from
univariate to multivariate polynomials, and developed a strengthened upper bound and new lower bound
for real trigonometric polynomials. The lower bound provides a new sufficient condition to certify
global positivity of a real multivariate trigonometric polynomial, and this condition motivated a
new method to design two-dimensional, multirate, perfect reconstruction filter banks. The
demonstration of this application in this paper is a preliminary study for the proposed filter bank
design algorithm; we plan to further investigate this design methodology, including extensions to
non-uniform and/or data-adaptive filter banks.

%%% Local Variables:
%%% mode: latex
%%% TeX-master: "trig_poly_main"
%%% End:

\label{sec:refs}
\bibliographystyle{myIEEEtran}
\bibliography{IEEEabrv,/home/luke/research/bib/jabref}

\end{document}